\documentclass[journel,draftcls,onecolumn,12pt,twoside]{IEEEtran}
\usepackage{amsmath,amssymb,amsfonts,amsthm,mathrsfs}
\usepackage{graphicx}
\usepackage{balance}
\usepackage{graphics}
\usepackage{graphicx}
\usepackage{epsfig}
\usepackage{algorithm}               
\usepackage{algorithmic}             
\usepackage{multirow}                
\usepackage{amsmath}
\usepackage{xcolor}
\usepackage{graphicx}
\usepackage{subfigure}
\usepackage{stfloats}
\usepackage{bm}
\usepackage{bbm}
\usepackage{fancyhdr}

\usepackage{amsthm}
\newtheorem{theorem}{Theorem}
\newtheorem{lemma}{Lemma}
\theoremstyle{definition}

\theoremstyle{remark}

\theoremstyle{corollary}

\begin{document}\pagestyle{empty}
 
 \IEEEoverridecommandlockouts
    
\title{User Preference Aware Lossless Data Compression at the Edge}

\author{
    Yawei~Lu,~\IEEEmembership{Student Member,~IEEE},~Wei~Chen,~\IEEEmembership{Senior Member,~IEEE},\\~and~H. Vincent Poor,~\IEEEmembership{Fellow,~IEEE}
    
    \thanks{Yawei Lu and Wei Chen are with the Department of Electronics Engineering, Tsinghua University, Beijing, 100084, China, e-mail: lyw15@mails.tsinghua.edu.cn, wchen@tsinghua.edu.cn. Yawei Lu is also with the Department of Electrical Engineering, Princeton University, New Jersey, 08544, USA, e-mail: yawil@princeton.edu.}

    \thanks{H. Vincent Poor is with the Department of Electrical Engineering, Princeton University, New Jersey, 08544, USA, e-mail: poor@princeton.edu.}
    
    \thanks{ This research was supported in part by the U.S. National Science Foundation under Grants CCF-093970 and CCF-1513915, the Chinese national program for special support for eminent professionals (10,000-talent program), Beijing Natural Science Foundation (4191001), and the National Natural Science Foundation of China under Grant No. 61671269. }

}
  
    \maketitle
    \pagestyle{fancy}  
    \thispagestyle{fancy} 
    \lhead{} 
    \chead{} 
    \rhead{\thepage} 
    \lfoot{} 
    \cfoot{} 
    \rfoot{} 
    \renewcommand{\headrulewidth}{0pt} 
    \renewcommand{\footrulewidth}{0pt} 
  
\maketitle\thispagestyle{empty}

\begin{abstract}
Data compression is an efficient technique to save data storage and transmission costs. However, traditional data compression methods always ignore the impact of user preferences on the statistical distributions of symbols transmitted over the links. Notice that the development of big data technologies and popularization of smart devices enable  analyses on user preferences based on data collected from personal handsets. This paper presents a user preference aware lossless data compression method, termed edge source coding, to compress data at the network edge. An optimization problem is formulated to minimize the expected number of bits needed to represent a requested content item in edge source coding.  For  edge source coding under discrete user preferences,  DCA (difference of convex functions programming algorithm) based and $k$-means++  based algorithms are proposed to give codebook designs. For  edge source coding under continuous user preferences, a sampling method is applied to give codebook designs. In addition, edge source coding is extended to the two-user case and codebooks are elaborately designed to utilize  multicasting opportunities.  Both theoretical analysis and simulations demonstrate the optimal codebook design should take into account  user preferences.
\end{abstract}

\begin{IEEEkeywords}
    Lossless data compression,  user preference, edge source coding, codebook design, difference of convex functions
\end{IEEEkeywords}
\IEEEpeerreviewmaketitle

\section{Introduction}
With the popularization of smart devices and the rise of mobile multimedia applications, network traffic has undergone an explosive growth in the past decades.  By using fewer bits to encode information than the original representation, data compression is an efficient  technique to reduce the cost of storing or transmitting data. Traditional data compression methods are always source-based. The major design criterion is the average compression ratio for content  generated by the information source. In the era of 5G, a large number of various content items are generated at the network.  Because each user might only be interested in a few kinds of content items, a source-based data compression algorithm probably cannot achieve a  satisfactory compression ratio for these specific kinds of content items. Recently the development of big data technologies and  popularization of smart devices enable  us to analyze user preferences and predict  user requests based on private data collected from personal handsets. It becomes possible to improve the efficiency of data compression according to  user preferences. 

Since C. E. Shannon published his famous source coding theorem \cite{shannon}, data compression has been widely studied and various data compression schemes have been proposed. The LZ77 algorithm was presented in \cite{lz77}, which is the  basis of several ubiquitous compression schemes, including ZIP and GIF. To improve the compression ratio, lossy compression algorithms  for various types of content have also been developed, such as MPEG-4 for videos \cite{mpeg4} and JPEG for images \cite{jpeg}. To solve the scalability problem resulting from packet classification in network services, a lossy compression based classifier was designed to reduce the classifier size \cite{rottenstreich}.  Recently considerable attention has been paid to data compression in sensor networks \cite{sensor1}-\cite{sensor3}. A lossy compression algorithm was proposed to cope with large volumes of  data generated by meters in smart distribution systems \cite{sensor1}. To reduce the power consumption of wireless sensors in Internet of Things  applications, \cite{sensor2} proposed a hybrid data compression scheme, in which both lossless and lossy techniques were used. For wireless sensor networks with correlated sparse sources, a complexity-constrained distributed variable-rate
quantized compression sensing method  was developed in \cite{sensor3}.  
Machine learning techniques were  also applied in  data compression and transmission \cite{source-coding-deep-learning-1}-\cite{source-coding-deep-learning-2}. A joint source-channel coding design of text was provided by deep learning in \cite{source-coding-deep-learning-1}. End-to-end communications were implemented by neural networks in \cite{source-coding-deep-learning-2}.

With the coming of Internet of Things and 5G, the edge of networks is attracting more and more attentions \cite{edge-1}, because it can handle the concerns of  latency requirement, bandwidth saving, and data security \cite{edge-2}.  Edge computing has the potential to support ``smart city" \cite{edge-3},  improve vehicle services \cite{edge-4}, and implement task offloading \cite{edge-5}. As edge nodes are closer to users, caching in the edge is capable of reducing the delivery latency and network congestion. Various content delivery and caching schemes were developed for edge networks \cite{edge-6}-\cite{lu2}. The energy efficiency of edge caching was revealed in \cite{edge-9}. In \cite{edge-10}, a learning-based method was proposed for edge caching. 
Furthermore, many papers have focused on improve the performance of edge caching via user preferences \cite{yanqi}-\cite{user-pref-1}.  In edge networks, user preferences also contribute to improve streaming
video services \cite{user-pref-2},   the Quality-of-Experience \cite{user-pref-3}, wireless resource allocations \cite{user-pref-4}, and device-to-device content deliveries \cite{user-pref-5}. 

In this paper, we are interested in data compression at the edge under user preferences.  More specifically, we consider the situation in which users are located at the network edge and are connected to  content providers through a service provider. In traditional communication systems,   data compression processes are executed at content providers when content items are generated.  In these compression processes, more common symbols should be assigned with shorter codewords to minimize the expected codeword length according to  information theory \cite{inf}. However, because   user interests vary with  content items,  a symbol that is common in the whole set of content items might not be common in the requested content items.  In other words, the statistical distribution of symbols in the transmission link is not likely to be identical to that in the information source due to user preferences.  As a result, it is necessary to re-compress content items at  the service provider according to user preferences. 
In this paper, a  user preference aware lossless data compression scheme, termed edge source coding,  is proposed to compress data according to finitely many codebooks at the network edge. 

To obtain the optimal codebooks, we formulate an optimization problem for edge source coding, which is however nonconvex for the general case.  For edge source coding with a single codebook, we solve the optimization problem via the method of Lagrange multipliers. An optimality condition is further presented. For  edge source coding under discrete user preferences, the optimization problem reduces to a clustering problem.  DCA (difference of convex functions programming algorithm) based and $k$-means++ based algorithms are proposed to give codebook designs \cite{$k$-means}. For  edge source coding under continuous user preferences, a sampling based method is applied to give codebook designs. We further extend edge source coding to the two-user case and present two algorithms to reduce transmission costs by using  multicasting opportunities. Finally, simulation results are presented to demonstrate the potential of edge source coding and the performance of our algorithms.

 The rest of the paper is organized as follows. Section \ref{section2} presents the content distribution system model and the concept of edge source coding. And then the formulation of edge source coding  is given in Section \ref{section3}. Sections \ref{single_user_discrete} presents codebook designs for edge source coding. In Section \ref{section6}, edge source coding for two users with preference correlation is studied. Simulation results are presented in Section \ref{section7}. Finally, Section \ref{section8} concludes this paper and lists some directions for future research.

\section{System Model}\label{section2}
In this section, we present our system model, based on which the key idea of edge source coding is introduced.

\subsection{Content Distribution System}\label{system_model}

 \begin{figure} 
    \centering
    \includegraphics[width=16cm]{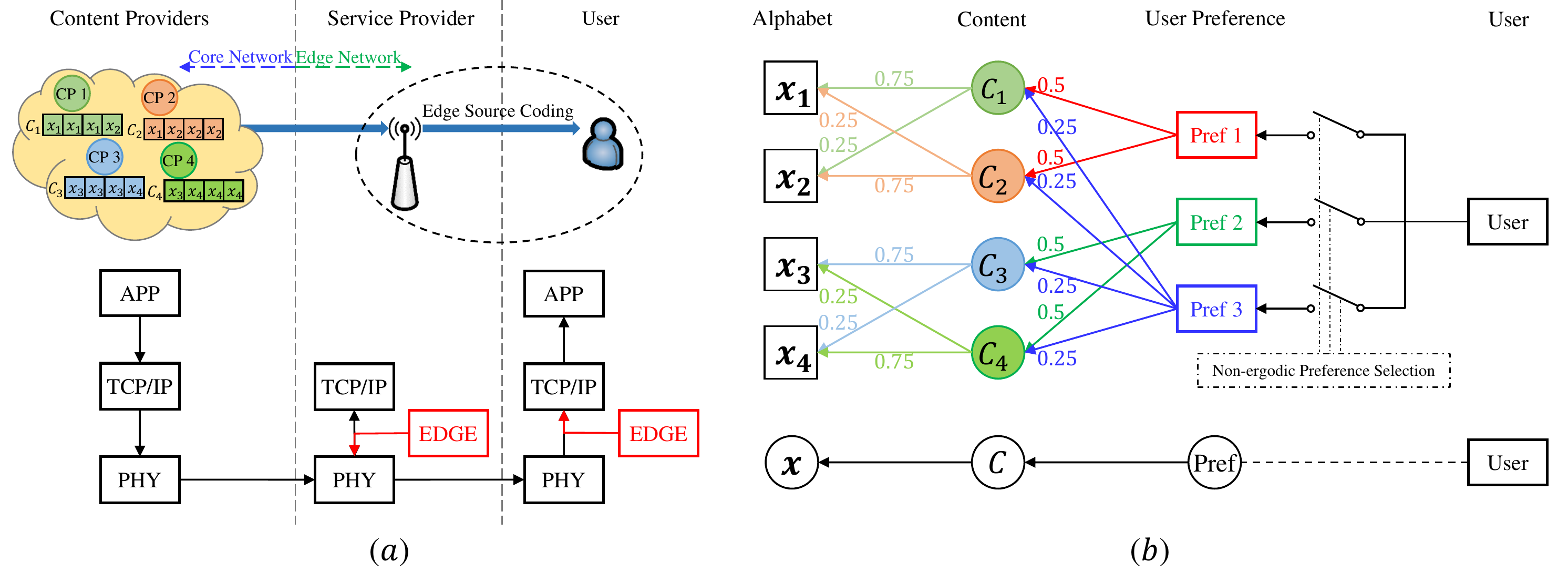}
    \caption{Content distribution system and a demo to illustrate how user preferences impact the statistical distributions of  symbols transmitted in the network edge. In this realization, the alphabet is $\mathcal{X}=\{x_1,x_2,x_3,x_4\}$. The four content items $C_1$, $C_2$, $C_3$, and $C_4$ generated by the four CPs are with SPVs $\bm{p}_1=[0.75,0.25,0,0], \bm{p}_2=[0.25,0.75,0,0],\bm{p}_3=[0,0,0.75,0.25],$ and $\bm{p}_4=[0,0,0.25,0.75]$, respectively. The user has three possible preferences. Under Pref 1, the symbols $x_3$ and $x_4$ never appear in the link from the service provider to the user.}    
    \label{CDS}  
\end{figure}

Consider a content distribution system as shown in Fig. \ref{CDS}.  A user is interested in content items generated by content providers (CPs) scattered around the Web. A base station (BS) serves as service provider to satisfy the user requests.\footnote{In this paper, the terms base station and service provider are used interchangeably. } Assume each CP  produces its content item by choosing symbols from the same alphabet $\mathcal{X}=\{x_1,...,x_N\}$ and each content item  consists of $L$ symbols.     The $L$ symbols of a content item are generated independently and based on the same discrete distribution, denoted by $\bm{p}=[p_1,...,p_N]$, where $p_n$ denotes the probability that a symbol is $x_n$.   The vector $\bm{p}$ is an attribute of that content item and will be referred to as the symbol probability vector (SPV). Note that a content item with SPV $\bm{p}$ has entropy $LH(\bm{p})=-L\sum_{n=1}^{N}p_n\log p_n$.\footnote{$H(\cdot)$ denotes the entropy and $\log$ represents the binary logarithm.} 
 All the feasible SPVs form a set $\Omega=\large \{[p_1,...,p_N]:\sum_{n=1}^{N}p_n=1\text{ and } p_n\ge 0 \text{ for } n \in [N]  \}$.\footnote{For a positive integer $N$, $[N]$ represents the set $\{1,2,...,N\}.$}

In the content distribution system, the user issues requests for the content items.  Let  $f(\bm{p})$ describe the interest of the user in the  content item with SPV $\bm{p}$. Then $f$ is a probability density function with support set $\Omega$.  The function $f$ will also be referred to as the user preference. To give further insights, we present a tripartite  graph model as shown in Fig. \ref{CDS}$(b)$.   As stated before, each content item has its own probability distribution of the symbols, i.e., SPV. The symbol distributions of two content items can be totally different. We assume that each user has an individual preference that can be characterized by the request probability for various content items. In practice, the empirical request probability can be learned from a user's historical requests. Further, we assume that a user has a fixed preference in our considered timescale. In other words, a non-ergodic preference selection is assumed to determine the preference of a user newly accessing to the service provider, or more particularly, a BS. In other words, the random symbols of the equivalent source at the edge are generated according to a simple probabilistic graphical model, as shown at the bottom of Fig. \ref{CDS}$(b)$.

 The BS responds to a user request and initiates an end-to-end transmission to satisfy it. To improve  transmission efficiency, data compression techniques can be used to  eliminate statistical redundancy of the original symbol sequences. 
Traditionally, data compression is  executed only in application (APP) layers at the content providers.  Content items are compressed according to the statistical distributions of symbols, i.e., their SPVs, and are associated with codebooks for decoding. Edge information including user preferences is always ignored in traditional data compression methods, which however have a significant impact on the statistical distributions of symbols transmitted at the edge, as illustrated in Fig. \ref{CDS}. To reduce the transmission cost at the network edge, we present a user preference aware lossless data compression method, termed  edge source coding, to re-compress original symbol sequences at the service provider. As shown in Fig. \ref{CDS}$(a)$, edge source coding exploits edge information in physical (PHY) layer transmissions.

\subsection{Edge Source Coding}\label{ESC}

In  edge source coding, the BS compresses the content items according to finitely many binary codebooks to satisfy the user requests. These binary codebooks are cached in both the BS and the user. In this way, the encoded symbol sequences do not need to be associated with a whole codebook for decoding. Let $K$ be the total number of codebooks used in edge source coding. In the $k$-th codebook, the symbol $x_n$ is represented by $l_{k,n}$ bits. If the BS applies the $k$-th codebook to encode a content item with SPV $\bm{p}$, this content item can be represented by $L\sum_{n=1}^{N}p_nl_{k,n}=L(H(\bm{p})+D(\bm{p}||\bm{q_k}))$ bits, where $\bm{q}_k=[q_{k,1},...,q_{k,n}]$ and $q_{k,n}=2^{-l_{k,n}}$.\footnote{$D(\cdot||\cdot)$  denotes the Kullback-Leibler divergence.} To satisfy the request for this content item, the BS only needs to transmit these $L\sum_{n=1}^{N}p_nl_{k,n}$ bits. The user can decode the received bits by trying all the cached codebooks. In the rest of this paper, we also use $\bm{q}_k$ to represent the $k$-th codebook.  According to  Kraft's inequality \cite{inf}, $\bm{q}_k$ should satisfy 
\begin{equation}\label{kraft}
\sum_{n=1}^{N}q_{k,n}\le 1
\end{equation}
to ensure decodability. 

As there are $K$ codebooks, the minimum cost to satisfy a request for the content item with SPV $\bm{p}$ is given by
$L(H(\bm{p})+\min_{k\in [K]} D(\bm{p}||\bm{q_k}))$. Because the SPV of the requested content item obeys a probability density function $f$, the transmission cost for a requested content item is given by 
\begin{equation}\label{trans_cost}
\begin{split}
\int_{\bm{p}\in\Omega}f(\bm{p})\left( L\left(H(\bm{p})+\min_{k\in [K]} D(\bm{p}||\bm{q_k})\right)\right) d\bm{p}.
\end{split}
\end{equation}
 In this paper, we aim to design  codebooks to minimize the transmission cost, i.e., expected number of bits to represent a requested content item.

\section{User Preference  Aware Compression: A Problem Formulation}\label{section3}
In this section, we formulate an optimization problem for  edge source coding  and solve it for the $K=1$ case.
In addition, an optimality condition for the optimization problem is presented.

As stated in Subsection \ref{ESC}, the codebooks can be described by $\{\bm{q}_k:k\in [K]\}$. Kraft's inequality is a sufficient  and necessary condition for the existence of a uniquely decodable code. We only require $\bm{q}_k$ to obey  Kraft's inequality and relax the constraint that each component of $\bm{q}_k$ should be a negative integer power of two. Note that Eq. (\ref{trans_cost}) can be rewritten as 
\begin{equation}\label{re_trans_cost}
\begin{split}
L\int_{\bm{p}\in\Omega}f(\bm{p})H(\bm{p})d\bm{p}+L\int_{\bm{p}\in\Omega}f(\bm{p})\min_{k\in [K]} D(\bm{p}||\bm{q_k})d\bm{p}
\end{split}
\end{equation}
To minimize the transmission cost, we only need to minimize $\int_{\bm{p}\in\Omega}f(\bm{p})\min_{k\in[K]} D(\bm{p}||\bm{q_k})d\bm{p}$ in Eq. (\ref{re_trans_cost}), which indicates the additional transmission cost per symbol under user preference $f$ due to the mismatch between  SPVs and  codebooks. 

Let $\Omega_k$ be the set of SPVs that  have the smallest Kullback-Leibler divergence with codebook $\bm{q}_k$ among all the codebooks, i.e.,
\begin{equation}\label{split_omega}
    \Omega_k = \{\bm{p}:\bm{p}\in \Omega, D(\bm{p}||\bm{q_k})\le D(\bm{p}||\bm{q_j}) \text{ for } j \in[K]\}.
\end{equation}
One can see that the sets $\Omega_k$  are disjoint, $\Omega=\cup_{k=1}^K\Omega_k$, and $\min_{j\in [K]} D(\bm{p}||\bm{q}_j) = D(\bm{p}||\bm{q}_k) $ for $\bm{p}\in \Omega_k$.\footnote{If a point $\bm{p}$ satisfies $D(\bm{p}||\bm{q}_{k_1})=D(\bm{p}||\bm{q}_{k_2})\le D(\bm{p}||\bm{q}_j) $ for $j\in [K],$ then $\bm{p}$ can be classified into $\Omega_{k_1}$ and $\Omega_{k_2}$ randomly to ensure that the sets $\Omega_k$  are disjoint.}  To minimize the transmission cost per symbol, a content item with SPV belonging to $\Omega_k$ should be encoded according to codebook $\bm{q_k}$.
The following optimization problem gives the optimal codebook design:
\begin{equation}\label{op_1}
\begin{split}
\min_{\bm{q}_k}\quad  &\sum_{k=1}^{K}\int_{\bm{p}\in \Omega_k}f(\bm{p})D(\bm{p}||\bm{q}_k)d\bm{p}\\
\text{s.t.}\quad & \text{Eqs. } (\ref{kraft}) \text{ and } (\ref{split_omega}),\\
& q_{k,n}\ge 0, k\in[K],n\in[N].
\end{split}
\end{equation}
In the objective function, the integral over $\Omega$ is calculated by summing up the integrals over its subsets $\Omega_k$.  Note that swapping the values of different $\bm{q}_k$ does not change the objective value of problem (\ref{op_1}). Thus, problem (\ref{op_1}) has more than one optimal solutions, which further implies the nonconvexity of problem (\ref{op_1}). 

Because $\Omega_k$ is a set depending on ${\bm{q}_k},$ it is nontrivial to solve problem (\ref{op_1}) for the general case. We first solve problem (\ref{op_1}) for $K=1$ to give greater insight.
In this case, there is only one codebook $\bm{q}_1$ and thus we denote  it  as $\bm{q}=[q_1,...,q_n]$ for simplicity. Theorem \ref{theo1} presents the optimal codebook design.
\begin{theorem}\label{theo1}
    For $K=1$, the optimal codebook for  edge source coding is given by \footnote{$\mathbb{E}\{\cdot\}$  represents the expectation of a random variable.} 
    \begin{equation}\label{opt_solution}
    q_n = \frac{\mathbb{E}\{p_n\}}{\sum_{j=1}^{N}\mathbb{E}\{p_j\}}.
    \end{equation}
 
\end{theorem}
\begin{proof}
Let us consider the Lagrange function
\begin{equation}\label{lagrangian}
L(\bm{q},\lambda)=\int_{\bm{p}\in \Omega}f(\bm{p})D(\bm{p}||\bm{q})d\bm{p} + \lambda(\bm{1}^T\bm{q}-1),
\end{equation}
where $\lambda$ is the Lagrange multiplier and $\bm{1}$ is an $N$-dimensional vector all of whose components are equal to 1. The partial derivative of $ L(\bm{q},\lambda)$ with respect to $q_n$ is given by
\begin{equation}
\begin{split}
\frac{\partial L(\bm{q},\lambda)}{\partial q_n} & =-\frac{1}{q_n}\int_{\bm{p}\in \Omega}f(\bm{p})p_nd\bm{p} + \lambda \\
& =\lambda - \frac{\mathbbm{E}\{p_n\}}{q_n}.
\end{split}
\end{equation}
By the optimality conditions   $\frac{\partial L(\bm{q},\lambda)}{\partial q_n} = 0 $ and $\bm{1}^T\bm{q}=1$, the optimal codebook can be derived as Eq. (\ref{opt_solution}).
\end{proof}

In the optimal codebook,  the codeword of $x_n$ consists of $-\log q_n$ bits. Eq. (\ref{opt_solution}) implies that the more frequently the requested content items contain a  symbol, the shorter the codeword corresponding to this symbol will be. To improve the data compression efficiency at the network edge,   user preferences must be taken into account  in designing the codebook. 

Based on Theorem \ref{theo1}, the following theorem gives an optimality condition for  any values of the number of  codebooks, i.e., $K$.
\begin{theorem}\label{coro_1}
   The optimal codebook design for edge source coding satisfies
    \begin{equation}\label{opt_cond}
        q_{k,n}=\frac{\mathbb{E}\{p_n|\bm{p}\in \Omega_k\}}{\sum_{j=1}^{N}\mathbb{E}\{p_j|\bm{p}\in \Omega_k\}},
    \end{equation}
    where $\Omega_k$ is defined in Eq. (\ref{split_omega}).    
\end{theorem}
\begin{proof}
    After splitting $\Omega$ into $K$ disjoint subsets $\Omega_k$, the codebook $\bm{q}_k$ is used to encode content items with SPV in $\Omega_k$. Applying Theorem \ref{theo1} on each set $\Omega_k$ yields Eq. (\ref{opt_cond}).
\end{proof}

Theorem \ref{coro_1} reveals the coupling between the optimal  $\bm{q}_k$ and $\Omega_k$.
Note that the inequality $D(\bm{p}||\bm{q_k})\le D(\bm{p}||\bm{q_j}) $ can be expanded as
\begin{equation}
\sum_{n=1}^{N}p_n\log \frac{q_{j,n}}{q_{k,n}}\le 0
\end{equation}
which is a linear inequality in the vector $\bm{p}$. Then $\Omega_k$ is a convex polytope  characterized by hyperplanes defined by $\{\bm{q}_k:k\in [K]\}$. The vector $\bm{q}_k$ can be given by conditional expectations over $\Omega_k.$  To some extend, $\bm{q}_k$ can be viewed as a central point in $\Omega_k$.

\section{Codebook Design for User Preference Aware Compression}\label{single_user_discrete}

In this section, we investigate codebook design for edge source coding under user preference $f$. If the user is interested only in finitely many content items, $f$ reduces to a probability mass function.  Then $f$ is referred to as a discrete user preference. If $f$ is a continuous probability density function over $\Omega$, $f$ is referred to as a continuous user preference. Codebook designs for both discrete and continuous user preferences are considered in this section. It will be shown that  codebook designs rely on  user preferences in both the two cases. 

\subsection{DCA based Codebook Design under Discrete User Preferences}\label{dca_single}

In this subsection, edge source coding under discrete user preferences is studied. Mathematically,  $f(\bm{p})$ is nonzero only at finite points. As a result,  the edge source coding problem becomes  a  clustering problem. A DCA (difference of convex functions programming algorithm) based algorithm is proposed to give a codebook design. 

Denote  the set of nonzero points of $f$ as $\Omega^d=\{\bm{p}_1,...,\bm{p}_J \}$. We have $\Omega^d\subseteq \Omega$ and $|\Omega^d|=J$. Let $f_j$ be the  probability that the content item with SPV $\bm{p}_j$ is requested. Then, we have $f_j=f(\bm{p}_j)$ and $\sum_{j=1}^{J}f_j=1$. In this case, problem (\ref{op_1}) becomes the following form:
\begin{equation}\label{op_2}
\begin{split}
\min_{\bm{q}_k}\quad  &\sum_{k=1}^{K}\sum_{j\in \Omega_k^d}f_jD(\bm{p}_j||\bm{q}_k)\\
\text{s.t.}\quad & \bm{1}^T\bm{q}_k\le 1, k\in [K], \hspace{2.31cm} (\ref{op_2}.a)\\
& q_{k,n}\ge 0, k\in[K],n\in[N], \hspace{1cm} (\ref{op_2}.b)
\end{split}
\end{equation}
where
\begin{equation}\label{clustering}
    \Omega_k^d=\{j:D(\bm{p}_j||\bm{q}_k) \le D(\bm{p}_j||\bm{q}_i) \text{ for } i \in[K] \}.
\end{equation}
Similarly,
a content item with SPV in $\Omega_k^d$ should be encoded by codebook $\bm{q}_k$ in order to minimize the transmission cost.

Problem (\ref{op_2}) can be viewed as  a clustering problem. The $J$ points in $\Omega^d$ are clustered into $K$ subsets.  In the discrete case, 
Eq. (\ref{opt_cond}) becomes
\begin{equation}\label{opt_cond2}
            q_{k,n}=\frac{\sum_{j\in \Omega_k^d}f_jp_{j,n}}{\sum_{i=1}^{N}\sum_{j\in \Omega_k^d}f_jp_{j,i}}.
\end{equation}
Again, Eq. (\ref{opt_cond2}) implies that the best codebook design should take into account the discrete user preferences.
Once the optimal clustering scheme is obtained, the optimal codebooks can be derived from Eq. (\ref{opt_cond2}). To solve problem (\ref{op_2}) is to find the optimal clustering scheme. However, there are around $K^J$ clustering schemes. It is computationally prohibitive to traverse all the possible ones. To give a suboptimal clustering scheme within affordable space and time costs, we transform problem (\ref{op_2}) into a DC (difference of convex functions) programming problem and present a DCA based method to give an appropriate clustering scheme.  The core idea behind the construction of the DC programming is  probabilistic clustering  in which the codebook used to encode each content item is randomly selected.

Lemma \ref{theo2} presents an equivalent problem for problem (\ref{op_2}) in which the variables $\Omega_k$ are removed.
\begin{lemma}\label{theo2}
    Problem (\ref{op_2}) is equivalent to the following problem:
    \begin{equation}\label{op_4}
    \begin{split}
    \min_{r_{j,k}}\quad  &\sum_{k=1}^{K}\sum_{j=1}^{J}\sum_{n=1}^{N}f_jp_{j,n}r_{j,k}\log \frac{\sum_{j=1}^J\sum_{i=1}^{N}f_jp_{j,i}r_{j,k}} {\sum_{j=1}^Jf_jp_{j,n}r_{j,k}}\\
    \text{s.t.}\quad & \sum_{k=1}^{K}r_{j,k}=1,j\in [J], \hspace{2cm} (\ref{op_4}.a)\\
    &r_{j,k}\ge 0, k=\in[K],j\in[J]. \hspace{1cm} (\ref{op_4}.b)
    \end{split}
    \end{equation}    
\end{lemma}
\begin{proof}
    The key to prove the equivalence between two optimization problems is to show that the optimal solution of one can be easily derived from the optimal solution of the other and vice-versa. We first show the transformation from the optimal solution of problem (\ref{op_2}) to that of problem (\ref{op_4}).
    Let us consider a probabilistic clustering scheme. For the content item with SPV $\bm{p}_j$, let $r_{j,k}$ be the probability that codebook $q_k$ is selected to encode it. Then problem (\ref{op_2}) can be rewritten as  
    \begin{equation}\label{op_3}
    \begin{split}
    \min_{\bm{q}_k,r_{j,k}}\quad  &\sum_{j=1}^{J}\sum_{k=1}^{K}f_jr_{j,k}D(\bm{p}_j||\bm{q}_k)\\
    \text{s.t.}\quad & \bm{1}^T\bm{q}_k\le 1, k\in[K],\hspace{4.43cm} (\ref{op_3}.a)\\
    & \sum_{k=1}^{K}r_{j,k}=1,j\in[J], \hspace{4.14cm} (\ref{op_3}.b)\\
    & q_{k,n},r_{j,k}\ge 0, k\in[K],n\in[N],j\in[J],\hspace{1cm} (\ref{op_3}.c)
    \end{split}
    \end{equation}
    where the constraint Eq. ($\ref{op_3}.b$) corresponds to the probability normalization. 
    
    The method of Lagrange multipliers can be used to simplify problem (\ref{op_3}). Let us consider the Lagrange function 
    \begin{equation}\label{lagrangian_2}
    L(\bm{q}_k,r_{j,k},\lambda_k,\mu_j)=\sum_{j=1}^{J}\sum_{k=1}^{K}f_jr_{j,k}D(\bm{p}_j||\bm{q}_k)+ \sum_{k=1}^{K}\lambda_k(\bm{1}^T\bm{q}_k-1) + \sum_{j=1}^{J}\mu_j \left(  \sum_{k=1}^{K}r_{j,k}-1\right).
    \end{equation}
    The partial derivative of $L(\bm{q}_k,r_{j,k},\lambda_k,\mu_j)$ with respect to $q_{k,n}$ is given by 
    \begin{equation}
    \begin{split}
    \frac{\partial L(\bm{q}_k,r_{j,k},\lambda_k,\mu_j)}{\partial q_{k,n}} & =-\frac{1}{q_{k,n}}\sum_{j=1}^{J}f_jp_{j,n}r_{j,k} + \lambda_k.
    \end{split}
    \end{equation}
    Again, we have 
    \begin{equation}\label{opt_cond3}
    q_{k,n}=\frac{\sum_{j=1}^Jf_jp_{j,n}r_{j,k}}{\sum_{j=1}^J\sum_{i=1}^{N}f_jp_{j,i}r_{j,k}}.
    \end{equation}  
    according to the optimality conditions $\frac{\partial L(\bm{q}_k,r_{j,k},\lambda_k,\mu_j)}{\partial q_{k,n}} =0$ and $\sum_{n=1}^{N}q_{k,n}=1$.  It can be seen that Eqs. (\ref{opt_cond3}) and (\ref{opt_cond2}) are very similar. If we impose each $r_{j,k}$ equal to 0 or 1, Eq. (\ref{opt_cond3}) reduces to Eq.  (\ref{opt_cond2}). That is because the deterministic clustering  is a special case of probabilistic clustering. 
    Substituting Eq. (\ref{opt_cond3}) into problem (\ref{op_3}) yields problem (\ref{op_4}). 
    
    Suppose $\{\bm{q}_k^*:k\in [K]\}$ is the optimal solution of problem (\ref{op_2}). We set 
    \begin{equation}
    r_{j,k}^{*} = \begin{cases}
    1, & \text{if } k = \arg\min_{k_0\in [K]} D(\bm{p}_j||\bm{q}_{k_0}^*),\\
    0, &  \text{if } k \neq  \arg\min_{k_0\in [K]} D(\bm{p}_j||\bm{q}_{k_0}^*).
    \end{cases}
    \end{equation}
    Then $\{\bm{q}_k^*\}$ and $\{r_{j,k}^*\}$ form the optimal solution of problem (\ref{op_3}) and Eq. (\ref{opt_cond3}) holds for $\{\bm{q}_k^*\}$ and $\{r_{j,k}^*\}$. As a result,  $\{r_{j,k}^*\}$ is the optimal solution of problem (\ref{op_4}). The transformation from 
    from the optimal solution of problem (\ref{op_4}) to that of problem (\ref{op_2}) is presented in Appendix.
\end{proof}
Note that $r_{j,k}$ represents that the probability that the content item with SPV $\bm{p}_j$ is encoded by codebook $\bm{q}_k$. Lemma \ref{theo2} indicates the best probabilistic clustering is exactly a deterministic clustering. Although the constraints of problem  (\ref{op_4}) are linear, problem (\ref{op_4}) is still intractable due to the nonconvex objective function. Notice the objective function of problem (\ref{op_4}) contains the logarithms of fractions. We can transform the original problem (\ref{op_2}) into a DC programming problem and apply DCA to solve it.

\begin{theorem}\label{theo2.1}
    Problem (\ref{op_2}) is equivalent to a DC programming problem having the following form:
    \begin{equation}\label{op_5}
    \begin{split}
    \min_{\bm{x}}\quad  &\sum_{m=1}^{M}\lambda_{1,m} x_m \log x_m - \sum_{m=1}^{M}\lambda_{2,m} x_m \log x_m\\
    \text{s.t.}\quad & \bm{A}_1\bm{x} = \bm{b}_1,\hspace{2.232cm} (\ref{op_5}.a) \\
    &x_m\ge 0, m\in[M], \hspace{1cm} (\ref{op_5}.b)
    \end{split}
    \end{equation} 
    where $M=K+KN+KJ$,  $\lambda_{1,m}=1$ for $1\le m \le M$,  $\lambda_{1,m}=0$ for other values of $m$,  $\lambda_{2,m} = 1$ for $K+1\le m\le K+ KN$, and $\lambda_{2,m} = 0$ for other values of $m$.
\end{theorem}
\begin{proof}
    According to Lemma \ref{theo2}, we only need to show the equivalence between problems (\ref{op_4}) and (\ref{op_5}).
    By introducing two groups of auxiliary variables 
    \begin{equation}\label{auxi1}
    s_{k,n}=\sum_{j=1}^Jf_jp_{j,n}r_{j,k},\hspace{0.9cm}
    \end{equation}  
    \begin{equation}\label{auxi2}
    t_k=\sum_{j=1}^J\sum_{i=1}^{N}f_jp_{j,i}r_{j,k},
    \end{equation}  
    and defining $\bm{x} = [t_1,..,t_K,s_{1,1},...,s_{K,N}, r_{1,1},...,r_{J,K}]$, problem (\ref{op_4}) can be rewritten (\ref{op_5}).  
\end{proof}

\begin{algorithm}[t]
    \caption{DCA based Codebook Design under Discrete User Preferences}
    \label{dca_clustering}
    \renewcommand{\algorithmicrequire}{ \textbf{Input:}} 
    \renewcommand{\algorithmicensure}{ \textbf{Output:}} 
    \begin{algorithmic}[1]
        \REQUIRE    $\bm{p}_1,...,\bm{p_J},f_1,...,f_J,\varepsilon$;        \\
        \ENSURE $\bm{q}_1,...,\bm{q}_K$;\\
        \STATE  Calculate $\bm{\lambda}_1,\bm{\lambda}_2,\bm{A}_1,\bm{b}_1$ according to $\bm{p}_1,...,\bm{p}_J$ and $f_1,...,f_J$;
        \STATE Solve problem (\ref{op_5}) according to Algorithm \ref{dca} and obtain $\bm{x}^{\text{opt}}$;
        \STATE Extract $r_{1,1},...,r_{J,K}$ from $\bm{x}^{\text{opt}}$;
        \FOR{$j =1$ to $J$}
        \STATE $k^{\max } = \arg \max_k r_{j,k}$;
        \STATE Set $r_{j,k^{\max} } = 1 $ and $r_{j,k} = 0 $ for $k\neq k^{\max }$;
        \ENDFOR
        \STATE Compute $\bm{q}_1,...,\bm{q}_K$ according to Eq. (\ref{opt_cond3}).
    \end{algorithmic}
\end{algorithm}

\begin{algorithm}[t]
    \caption{DCA for Problem (\ref{op_5})}
    \label{dca}
    \renewcommand{\algorithmicrequire}{ \textbf{Input:}} 
    \renewcommand{\algorithmicensure}{ \textbf{Output:}} 
    \begin{algorithmic}[1]
        \REQUIRE    $\bm{\lambda}_1,\bm{\lambda}_2,\bm{A}_1,\bm{b}_1,\varepsilon$;        \\
        \ENSURE $\bm{x}^{\text{opt}}$;\\
        \STATE Randomly  initialize $\bm{x}^{\text{old}}$ satisfying $\bm{A}_1\bm{x}^{\text{old}}=\bm{b}_1$;
        \STATE Calculate $R^{\text{new}}=\sum_{m=1}^{M}(\lambda_{1,m}-\lambda_{2,m}) x_m^{\text{old}} \log x_m^{\text{old}} $;
        \STATE Initialize $R^{\text{old}}=\inf$;
        \WHILE{$|R^{\text{old}}-R^{\text{new}}|> \varepsilon$}
        \STATE $R^{\text{old}}=R^{\text{new}}$;
        \FOR{$m=1$ to $M$}
        \STATE $y_{m}^{\text{old}} = \lambda_{2,m} (1 + \log (x_m^{\text{old}}))$;
        \ENDFOR
        \STATE Solve problem (\ref{op_6}) and obtain $\bm{x}^{\text{new}}$;
        \STATE Calculate $R^{\text{new}}=\sum_{m=1}^{M}(\lambda_{1,m}-\lambda_{2,m}) x_m^{\text{new}} \log x_m^{\text{new}} $;
        \STATE $x^{\text{old}}=x^{\text{new}}$;
        \ENDWHILE
        \STATE $\bm{x}^{\text{opt}}=\bm{x}^{\text{new}}$.
    \end{algorithmic}
\end{algorithm}

Note that $s_{k,n}$ is weighted average of SPVs. The weights are related to the probabilities in the probabilistic clustering scheme. In problem (\ref{op_5}), the linear equality constraint $\bm{A}_1\bm{x} = \bm{b}_1$ results from Eqs. ($\ref{op_4}.a$), (\ref{auxi1}), and (\ref{auxi2}). Denote $\bm{\lambda}_1=[\lambda_{1,1},...,\lambda_{1,M}]$ and $\bm{\lambda}_2=[\lambda_{2,1},...,\lambda_{2,M}]$, which are two 0-1 vectors. 
Algorithm \ref{dca_clustering} provides a codebook design based on DCA for problem (\ref{op_5}). Considering Lemma \ref{theo2} implies  that each $r_{j,k}$ is 0-1 in the optimal solution,  Algorithm \ref{dca_clustering} resets the values of $r_{j,k}$ in Steps 4-7 after obtaining $\bm{x}^{\text{opt}}$ from Algorithm \ref{dca}. Algorithm \ref{dca} provides a suboptimal solution for problem (\ref{op_5}) by solving the following convex problem iteratively:
\begin{equation}\label{op_6}
\begin{split}
\min_{\bm{x}}\quad  &\sum_{m=1}^{M}\lambda_{1,m} x_m \log x_m - \left( \bm{y}^{\text{old}} \right) ^T\bm{x}   \\
\text{s.t.}\quad & \bm{A}_1\bm{x} = \bm{b}_1, \hspace{2.18cm} (\ref{op_6}.a)\\
&x_m\ge 0, m\in[M].\hspace{1cm} (\ref{op_6}.b)
\end{split}
\end{equation}
In each iteration, $\bm{y}^{\text{old}}$ is calculated in Steps 6-8. It can be seen that $\bm{y}^{\text{old}}$ is the gradient of function $\sum_{m=1}^{M}\lambda_{2,m} x_m \log x_m$ at the point $\bm{x}^{\text{old}}$. Then $\left( \bm{y}^{\text{old}} \right) ^T\bm{x} + C$ is a local approximation of function $\sum_{m=1}^{M}\lambda_{2,m} x_m \log x_m$ ($C$ is a constant). This is the reason we use $\left( \bm{y}^{\text{old}} \right) ^T\bm{x} $ to replace $\sum_{m=1}^{M}\mu_m x_m \log x_m$ in problem (\ref{op_6}). It should be pointed out that it is easy to initialize  $\bm{x}^{\text{old}}$ satisfying $\bm{A}_1\bm{x}^{\text{old}}=\bm{b}_1$. We only need to initialize $r_{j,k}$ satisfying  Eqs. ($\ref{op_4}.a$) and ($\ref{op_4}.b$)  and then generate $s_{k,n}$ and $t_k$ according to Eqs. (\ref{auxi1}) and (\ref{auxi2}).

\subsection{$k$-means++ based Codebook Design under Discrete User Preferences}\label{k_means_single}

\begin{algorithm}[t]
    \caption{ $k$-means++ based Codebook Design under Discrete User Preferences}
    \label{algo1}
    \renewcommand{\algorithmicrequire}{ \textbf{Input:}} 
    \renewcommand{\algorithmicensure}{ \textbf{Output:}} 
    \begin{algorithmic}[1]
        \REQUIRE    $\bm{p}_1,...,\bm{p_J},f_1,...,f_J$;        \\
        \ENSURE $\bm{q}_1,...,\bm{q}_K$;\\
        \STATE  Randomly select a vector from $\Omega^d$ as $\bm{q}_1$ where $\bm{p}_j$ is selected with probability $f_j$;
        \FOR{$k=2$ to $K$}
        \STATE Compute $G_j=\min_{i<k}D(\bm{p}_j||\bm{q}_i)$ for $j=1,...,J$;
        \STATE Randomly select a vector from $\Omega^d$ as $\bm{q}_k$ where $\bm{p}_j$ is selected with probability proportional to $f_jG_j^2$;
        \ENDFOR
        \STATE Initialize $\Omega_k^d$ according to Eq. (\ref{clustering});
        \WHILE {Eq. (\ref{opt_cond2}) is FALSE for some $n,k$}
        \STATE  Update $\bm{q}_k$ according to Eq. (\ref{opt_cond2});
        \STATE Update $\Omega_k$ according to Eq. (\ref{clustering});
        \ENDWHILE
        \RETURN $\{\bm{q}_k:k=1,...,K\}.$
    \end{algorithmic}
\end{algorithm}

Recall that the edge source coding problem reduces to a clustering problem under discrete user preferences. Considering that $k$-means++ is a typical heuristic algorithm for clustering problems \cite{$k$-means},  we present a codebook design for edge source coding under discrete user preferences based on  the $k$-means++ approach in this subsection, as detailed in Algorithm 1. In Steps 1-5, each codebook is initialized according to the probabilities of the SPVs and the Kullback-Leibler divergence with  codebooks having been determined. In Steps 6-10, a variant of the  $k$-means approach is employed to update the clustering centers $\{\bm{q}_k:k\in[K]\}$. Compared with the traditional $k$-means++ approach, Algorithm \ref{algo1} uses the  Kullback-Leibler divergence instead of Euclidean distance to reassign the points into different clusters. In addition, the center of a cluster is derived from some conditional expectations in Algorithm \ref{algo1}, which is usually not the arithmetic mean of points in this cluster. 


\begin{algorithm}[t]
    \caption{$k$-means++ based Codebook Design under Continuous User Preferences}
    \label{algo2}
    \renewcommand{\algorithmicrequire}{ \textbf{Input:}} 
    \renewcommand{\algorithmicensure}{ \textbf{Output:}} 
    \begin{algorithmic}[1]
        \REQUIRE    $f,\varepsilon$;        \\
        \ENSURE $\bm{q}_1,...,\bm{q}_K$;\\
        \STATE  Randomly select a vector from $\Omega$ as $\bm{q}_1$ according to probability density function $f$;
        \FOR{$k=2$ to $K$}
        \STATE  $g(p)=\min_{i<k}D(\bm{p}||\bm{q}_i)$;
        \STATE Randomly select a vector from $\Omega$ as $\bm{q}_k$ according to $\frac{fg^2}{\int_{\bm{p}\in \Omega}fg^2d\bm{p}}$;
        \ENDFOR
        \STATE Initialize $\Omega_k$ according to Eq. (\ref{split_omega});
        \WHILE {$\min_{k,n}\left| q_{k,n} - \frac{\mathbb{E}\{p_n|\bm{p}\in \Omega_k\}}{\sum_{j=1}^{N}\mathbb{E}\{p_j|\bm{p}\in \Omega_k\}} \right| > \varepsilon $}
        \STATE  Update $\bm{q}_k$ according to Eq. (\ref{opt_cond});
        \STATE Update $\Omega_k$ according to Eq. (\ref{split_omega});
        \ENDWHILE
        \RETURN $\{\bm{q}_k:k\in [K]\}.$
    \end{algorithmic}
\end{algorithm}

\subsection{$k$-means++ based Codebook Design under Continuous User Preferences}\label{iter_algo}

In this subsection, we extend the $k$-means++ based algorithm proposed in the previous subsection to the continuous user preferences. Based on the coupling relationship between $\Omega_k$ and $\bm{q}_k$ revealed in Eqs. (\ref{split_omega}) and (\ref{opt_cond}), Algorithm \ref{algo2} presents an iterative method to obtain a suboptimal solution of problem (\ref{op_1}), which is a continuous version of Algorithm \ref{algo1}. In Algorithm \ref{algo2}, the codebooks are also initialized based on  user preferences. The parameter $\varepsilon$ in Step 7 is a sufficiently small positive number. Instead of summation, Algorithm \ref{algo2} calculates several integrals to update $\bm{q}_k$. The convergence of Algorithm \ref{algo2} is guaranteed by the fact that each update of $\bm{q}_k$ or $\Omega_k$ achieves a lower objective value  of problem (\ref{op_1}). Thus, Algorithm \ref{algo2} at least reaches a locally optimal point.

\begin{figure} 
    \centering
    \includegraphics[width=13cm]{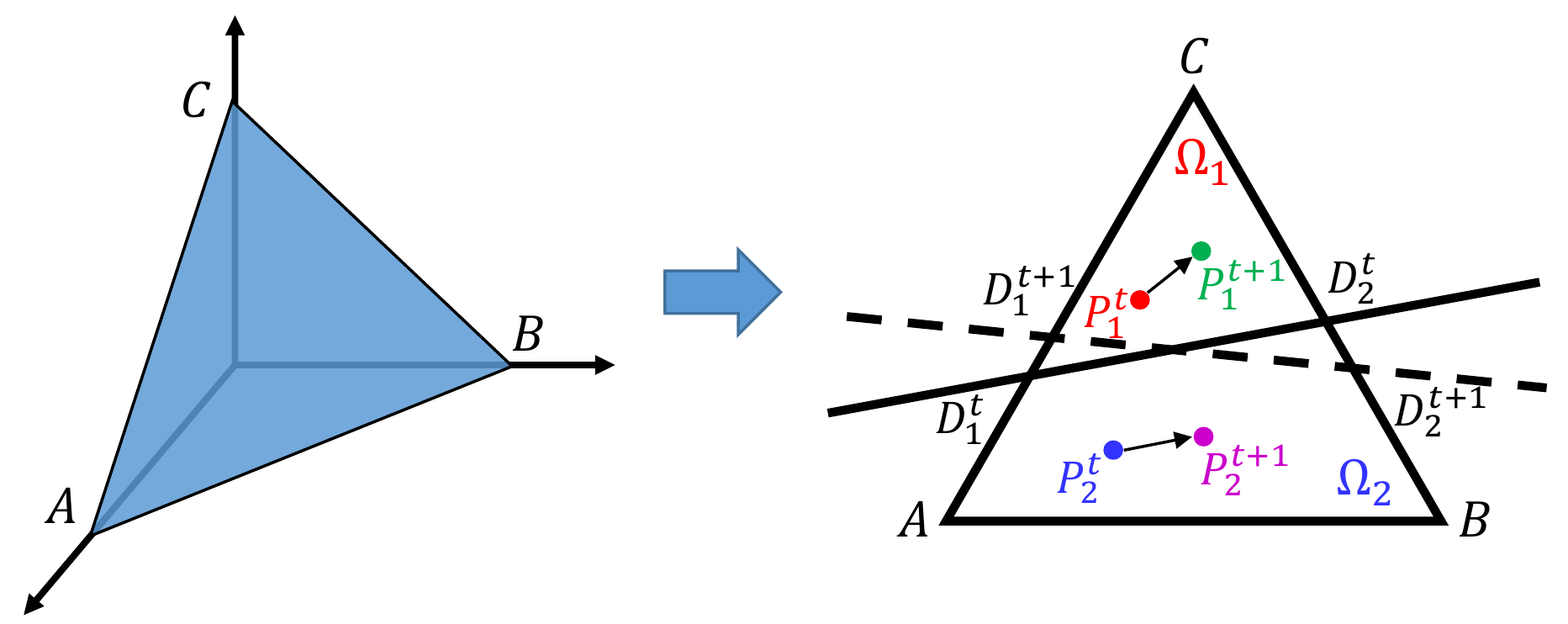}
    \caption{Illustration of the iterative process for $N=3$ and $K=2$.}    
    \label{n=3k=2}  
\end{figure}

We illustrate the iterative process in Algorithm \ref{algo2} by the case $N=3$, $K=2$, and the user preferences are uniform, i.e., $f=1/\int_{\bm{p}\in \Omega}d\bm{p}$. In this case, the set $\Omega$ is formed by points in an equilateral triangle with vertices $A=[1,0,0],B=[0,1,0],$ and $C=[0,0,1]$, as shown in Fig. \ref{n=3k=2}. The codebooks $\bm{q}_1$ and $\bm{q}_2$ are two three-dimensional vectors, corresponding to $P_1^t$ and $P_2^t$ respectively after the $t$-th iteration. Because there are only two codebooks, we need just a single hyperplane to split $\Omega$. The hyperplane is a line in the $N=3$ case, denoted as $D_1^tD_2^t$. Note that $D_1^tD_2^t$ is not the perpendicular bisector of $P_1^tP_2^t$ but instead is given by the equality $\sum_{n=1}^{3}p_n(\log q_{1,n} - \log q_{2,n}) =0$. In the $(t+1)$-th iteration, the points corresponding to the two codebooks are updated to  $P_1^{t+1}$ and $P_2^{t+1}$, which happen to be the centroids of the triangle $CD_1^tD_2^t$ and quadrilateral $AD_1^tD_2^tB$. When the user preferences are uniform, $P_1^{t+1}$ can be obtained by
\begin{eqnarray}
    P_1^{t+1} &=& \frac{C+D_1^t+D_2^t}{3}.
\end{eqnarray}
$P_2^{t+1}$ can also be obtained by a short calculation \cite{centroid}.
After obtaining $P_1^{t+1}$ and $P_2^{t+1}$, the new split line $D_1^{t+1}D_2^{t+1}$ can be calculated and further $\Omega_1$ and $\Omega_2$ are updated.

In the general case, namely, $f$ is an arbitrary probability density function and $N\ge 3$, it is intractable to calculate integrals over $\Omega_k$. This is because $\Omega_k$ is a convex polytope bounded by high dimensional hyperplanes. As a result, Algorithm \ref{algo2} is of high computational complexity for arbitrary $f$ and large $N$. More specifically,  the integrals over sets $\Omega_k$  induce the majority of the computational complexity. In the following subsection, we overcome this by a  sampling method.

\subsection{Sampling based Codebook Design under Continuous User Preferences}

\begin{algorithm}[t]
        \caption{Sampling based Codebook Design under Continuous User Preferences}
    \label{algo3}
    \renewcommand{\algorithmicrequire}{ \textbf{Input:}} 
    \renewcommand{\algorithmicensure}{ \textbf{Output:}} 
    \begin{algorithmic}[1]
        \REQUIRE    $f,S$;        \\
        \ENSURE $\bm{q}_1,...,\bm{q}_K$;\\
        \STATE  Sample $S$ points from $\Omega$ as $\bm{p}_1,...,\bm{p}_S$ according to probability density function $f$;
        \STATE Set $f_s = \frac{1}{S}$ for $s=1,...,S;$
        \STATE Run Algorithm \ref{algo1} or Algorithm $\ref{dca_clustering}$ with $\bm{p}_1,..,\bm{p}_S,f_1,...,f_S$ as input and obtain $\bm{q}_1,...,\bm{q}_K$;
        \RETURN $\{\bm{q}_k:k\in[K]\}.$
    \end{algorithmic}
\end{algorithm}

Algorithm \ref{algo3} presents a sampling based method to give a codebook design for edge source coding under continuous user preferences. The core idea behind Algorithm \ref{algo3} is that integrals over a set with a probability measure can be estimated by summations over sample points. In Algorithm \ref{algo3}, $S$ is the number of sample points. These points are sampled according to the function $f$. In Step 2, every sample point is associated with an identical probability $f_s=\frac{1}{S}$. Step 3 calls Algorithm \ref{algo1} or Algorithm \ref{dca_clustering} to gives a suboptimal codebook design.
 
 The number of sample points $S$ is a key parameter in Algorithm \ref{algo3}. On the one hand, too large $S$ will result in a high computing time. On the other hand, too small $S$ will reduce the estimation accuracy of the sampling method. There are two versions of Algorithm \ref{algo3} according to the algorithm called in Step 3. Simulations will demonstrate that these two versions of Algorithm \ref{algo3} have similar performance.  

 \section{Edge Source Coding for Two Users with Common Interests}\label{section6}
 In this section, we extend edge source coding to the two-user case, which will be referred to as user 1 and user 2. In contrast with the scenario considered in previous sections, two-user edge source coding is capable of reducing the total transmission cost by taking advantage of multicasting opportunities. The preferences of the two users will be described by a matrix. 
 
     Recall that $\Omega =\large \{[p_1,...,p_N]:\sum_{n=1}^{N}p_n=1\text{ and } p_n\ge 0 \text{ for } n = 1,...,N  \}$ represents the set of feasible SPVs. Let $\bm{p}[1]$ and $\bm{p}[2]$ denote the SPVs of the content items requested by user 1 and user 2, respectively. Then $\bm{p}[1]$ and $\bm{p}[2]$ are two random variables with support set   $\Omega \times \Omega$. We denote the probability density function of $\bm{p}[1]$ and $\bm{p}[2]$ as $f(\bm{p}[1],\bm{p}[2])$.  Let $K_1^t$ and $K_2^t$ be the total numbers of codebooks used by user 1 and user 2, respectively. As the two users probably request the same  content item, we assume there are $K_0$  codebooks that the two users have in common, denoted by $\bm{q}_1^0,...,\bm{q}_{K_0}^0$. In other words, user 1 and user 2 have $K_1=K_1^t-K_0$ and $K_2=K_2^t-K_0$ exclusive codebooks, respectively. Let $\{\bm{q}_1^1,...,\bm{q}_{K_1}^1\}$ and $\{\bm{q}_1^2,...,\bm{q}_{K_2}^2\}$ denote the sets of exclusive codebooks for user 1 and user 2, respectively.
 
 In the two-user edge source coding, multicasting opportunities can be utilized when the two users request the same content item. In this case, the requested content item will be encoded by a common codebook of the $K_0$ ones and then be transmitted to the two users simultaneously. If the two users request different content items, the BS has to compress the content item for each user separately. The total transmission cost to satisfy the user requests is given by 
 \begin{equation}\label{trans_cost_2}
 \begin{split}
 &\int_{\bm{p}[1] = \bm{p}[2] }\!f(\bm{p}[1],\bm{p}[2])\!\left(\! L\!\left(\!H(\bm{p}[1])\!+\!\min_{k\in [K_0]}\! D(\bm{p}[1]||\bm{q}_k^0)\!\right)\!\right)\!d\bm{p}[1]d\bm{p}[2] + \\
 & \int_{\bm{p}[1] \neq  \bm{p}[2]  }\!f(\bm{p}[1],\bm{p}[2])\!\left(\! L\!\left(\!H(\bm{p}[1])\!+\!\min\! \left\lbrace \min_{k\in [K_0]} \!D(\bm{p}[1]||\bm{q}_k^0), \min_{k\in [K_1]}\!D(\bm{p}[1]||\bm{q}_k^1) \!\right\rbrace \!\right)\!\right)\! d\bm{p}[1]d\bm{p}[2] +\\
  & \int_{\bm{p}[1] \neq  \bm{p}[2]  }\!f(\bm{p}[1],\bm{p}[2])\!\left(\! L\!\left(\!H(\bm{p}[2])\!+\!\min\! \left\lbrace \min_{k\in [K_0]} \!D(\bm{p}[2]||\bm{q}_k^0), \min_{k\in [K_2]}\!D(\bm{p}[2]||\bm{q}_k^2) \!\right\rbrace \!\right)\!\right)\! d\bm{p}[1]d\bm{p}[2].
 \end{split}
 \end{equation}
 The minimum operations in Eq. (\ref{trans_cost_2}) imply that each content item is encoded by the  codebook with the smallest Kullback-Leibler divergence. 
 The first term  in Eq. (\ref{trans_cost_2}) corresponds to the transmission cost when the two users' requests are identical and multicasting technique is used. The second and third terms correspond to the transmission costs of user 1 and user 2 when the two users' requests are different. 
 In this section, we aim to minimize Eq. (\ref{trans_cost_2}) by elaborately designing $K_0$ and codebooks  $\bm{q}_1^0,...,\bm{q}_{K_0}^0, \bm{q}_1^1,...,\bm{q}_{K_1}^1, \bm{q}_1^2,...,\bm{q}_{K_2}^2$. It should be noted that if the probability that the two users request the same content item is zero, i.e., $\text{Pr}\left( \bm{p}[1]=\bm{p}[2]\right) =0$, the first term in Eq. (\ref{trans_cost_2}) becomes 0.\footnote{$\text{Pr}(\cdot)$ denotes the probability of an event.}  Then multicasting opportunities cannot be created.  The optimal coding scheme should be setting $K_0=0$ and designing codebooks for the two users separately, which reduces to the problem considered in the previous sections. Thus, we pay attention to the situation that $\text{Pr}\left( \bm{p}[1]=\bm{p}[2]\right) \neq 0$ in this section.
 
 To make sure  $\text{Pr}\left( \bm{p}[1]=\bm{p}[2]\right) \neq 0$, $f(\cdot,\cdot)$ cannot be a continuous probability density function, because otherwise the Lebasgue measure of  the set $\{(\bm{p_1},\bm{p_2})\in \Omega \times \Omega :\bm{p_1}=\bm{p_2} )\}$ is 0. As discussed in Subsection \ref{dca_single}, we consider $f(\cdot,\cdot)$ to be a probability mass function. More specifically, user 1 and user 2 are interested only in $J$ different content items with SPVs $\bm{p}_1,...,\bm{p}_J\in \Omega^d$. The support of $f(\cdot,\cdot)$ reduces to $\{\bm{p}_1,...,
 \bm{p}_J\}\times \{\bm{p}_1,...,
 \bm{p}_J\}$, which allows us to describe the user preferences by a matrix $\bm{F}=(f_{i,j})_{J\times J}$ where $f_{i,j} = \text{Pr}(\bm{p}[1] = \bm{p}_i,\bm{p}[2] = \bm{p}_j)$.  The matrix $\bm{F}$ will be referred to as the joint preference matrix. The trace of $\bm{F}$ reflects the two users' common interests.  Hence, our task becomes minimizing Eq. (\ref{trans_cost_2}) given $\bm{F}$ and $K_1^t,K_2^t$. To this end, we formulate an  optimization problem  as follows:
\begin{equation}\label{op_7}
\begin{split}
\min_{\bm{q}_k^0,\bm{q}_k^1,\bm{q}_k^2}\quad  &\sum_{j=1}^{J}f_{j,j}\min_{k\in [K_0]} D(\bm{p}_j||\bm{q}_k^0)+\sum_{j=1}^{J}\sum_{i\neq j}f_{j,i}\min \left\lbrace \min_{k\in [K_0]} D(\bm{p}_j||\bm{q}_k^0), \min_{k\in [K_1]} D(\bm{p}_j||\bm{q}_k^1) \right\rbrace\\
& +\sum_{j=1}^{J}\sum_{i\neq j}f_{i,j}\min \left\lbrace \min_{k\in [K_0]} D(\bm{p}_j||\bm{q}_k^0), \min_{k\in [K_1]} D(\bm{p}_j||\bm{q}_k^2) \right\rbrace\\
\text{s.t.}\quad & \bm{1}^T\bm{q}_k^0\le 1,k \in[K_0],\hspace{3.63cm} (\ref{op_7}.a)\\
& \bm{1}^T\bm{q}_k^1\le 1, k\in[K_1],\hspace{3.63cm} (\ref{op_7}.b)\\
&\bm{1}^T\bm{q}_k^2\le 1, k\in[K_2],\hspace{3.63cm} (\ref{op_7}.c)\\
& K_0 + K _ 1 = K_1^t,K_0 + K_2 = K_2^t,\hspace{1.53cm} (\ref{op_7}.d)\\
& q_{k,n}^0, q_{k,n}^1, q_{k,n}^2\ge 0, K_0,K_1,K_2\in \mathbb{N},\hspace{1cm} (\ref{op_7}.e)\\
\end{split}
\end{equation} 
which is a mixed integer programming (MIP) with linear constraints.
The objective  of problem (\ref{op_7}) is derived from Eq. (\ref{trans_cost_2}) by removing a constant factor and a constant addend. Suffering from the nonconvex objective, problem (\ref{op_7}) is intractable. Two low-complexity algorithms are proposed to give   codebook designs. 

\subsection{DCA based Codebook Design  for Edge Source Coding with Two Users}

In this subsection, a DCA based algorithm is proposed to give a codebook design for edge source coding with two users. The major obstacle to solve problem (\ref{op_7}) results from the minimum operations, which impose that each content item is encoded by the codebook bringing the smallest transmission cost. Again, we consider a probabilistic clustering scheme, in which content items are encoded by randomly chosen codebooks. 

Let $r^0_{j,k}$ be the probability that codebook $\bm{q}_k^0$ is used when the two users request the content item with SPV $\bm{p}_j$ simultaneously. In the case only user 1 requests the content item with SPV $\bm{p}_j$, we let $r_{j,k}^{1}$ and $r_{j,k+K_0}^{1}$ be the probabilities that codebooks $\bm{q}_k^0$ and $\bm{q}_k^1$  are used, respectively. Similarly, we   define  $r_{j,k}^{2}$ and $r_{j,k+K_0}^{2}$. In the probability clustering scheme, problem (\ref{op_7}) becomes 
\begin{equation}\label{op_8}
\begin{split}
\min &\sum_{j=1}^{J}\sum_{k=1}^{K_0}f_{j,j}r^0_{j,k} D(\bm{p}_j||\bm{q}_k^0)+\sum_{j=1}^{J}\sum_{i\neq j}f_{j,i} \left( \sum_{k=1}^{K_0}r_{j,k}^{1} D(\bm{p}_j||\bm{q}_k^0)+\sum _{k=1}^{K_1}r^1_{j,k+K_0} D(\bm{p}_j||\bm{q}_k^1) \right)\\
& +\sum_{j=1}^{J}\sum_{i\neq j}f_{i,j} \left( \sum_{k=1}^{K_0}r_{j,k}^{c,2} D(\bm{p}_j||\bm{q}_k^0)+\sum _{k=1}^{K_2}r^2_{j,k+K_0} D(\bm{p}_j||\bm{q}_k^2) \right)\\
\text{s.t.}\,& \text{Constraints of Problem (\ref{op_7})},\\
& \sum_{k=1}^{K_0 }r_{j,k}^0 = 1, j \in[J],\hspace{2.3cm} (\ref{op_8}.a)\\
&  \sum_{k=1}^{K_1^t} r_{j,k}^{1}  = 1,  j \in[J],\hspace{2.3cm} (\ref{op_8}.b)\\
&  \sum_{k=1}^{K_2^t} r_{j,k}^{2}  = 1, j \in[J],\hspace{2.3cm} (\ref{op_8}.c)\\
& r_{j,k}^0,r_{j,k}^1,r_{j,k}^2\ge 0. \hspace{3cm} (\ref{op_8}.d)
\end{split}
\end{equation} 
Similar to the analysis in Subsection \ref{dca_single},
problem (\ref{op_8}) is equivalent to problem (\ref{op_7}). The probabilistic clustering method eliminates the minimum operations without any loss in the optimality. 

If we  fix $K_0$, problem (\ref{op_8}) becomes an optimization problem without integer variables. Then, the method of Lagrange multipliers can help us get greater insight on problem (\ref{op_8}). We have the theorem stated below. 
\begin{theorem}
    The optimal probabilistic clustering satisfies 
    \begin{eqnarray}
    \label{opt_cond_3}
    q_{k,n}^0&=&\frac{\sum_{j=1}^J\left(  r^0_{j,k}f_{j,j}+ r^{1}_{j,k}\sum_{i\neq j}f_{j,i} +r^{2}_{j,k}\sum_{i\neq j}f_{i,j}  \right) p_{j,n} }{\sum_{n_0=1}^N\sum_{j=1}^J\left(  r^0_{j,k}f_{j,j}+ r^{1}_{j,k}\sum_{i\neq j}f_{j,i} +r^{2}_{j,k}\sum_{i\neq j}f_{i,j}  \right) p_{j,n_0}},\\
    \label{opt_cond_4}
    q_{k,n}^1&=&\frac{\sum_{j=1}^J\left(  \sum_{i\neq j}f_{j,i}   \right) r^{1}_{j,k+K_0}p_{j,n} }{\sum_{n_0=1}^N\sum_{j=1}^J\left(  \sum_{i\neq j}f_{j,i}   \right) r^{1}_{j,k+K_0}p_{j,n_0}},\\
    \label{opt_cond_5}
    q_{k,n}^2&=&\frac{\sum_{j=1}^J\left(  \sum_{i\neq j}f_{i,j}   \right) r^{2}_{j,k+K_0}p_{j,n} }{\sum_{n_0=1}^N\sum_{j=1}^J\left(  \sum_{i\neq j}f_{i,j}   \right) r^{2}_{j,k+K_0}p_{j,n_0}}.
    \end{eqnarray}
\end{theorem}
\begin{proof}
    Taking partial derivatives of the Lagrange function of problem (\ref{op_8}) yields Eqs. (\ref{opt_cond_3})-(\ref{opt_cond_5}). 
\end{proof}

Eqs. (\ref{opt_cond_3})-(\ref{opt_cond_5}) enable us to remove the optimization variables $\bm{q}_k^0,\bm{q}_k^1,\bm{q}_k^2$ in problem (\ref{op_8}). As a consequence, problem (\ref{op_8}) becomes an optimization problem similar to problem (\ref{op_4}), which contains the logarithms of fractions in the objective. As discussed in Subsection \ref{dca_single}, we introduce several auxiliary variables to obtain a simplified equivalent problem of problem (\ref{op_8}):
\begin{eqnarray}
\label{auxi3}
s^0_{k,n}&=&\sum_{j=1}^J\left(  r^0_{j,k}f_{j,j}+ r^{1}_{j,k}\sum_{i\neq j}f_{j,i} +r^{2}_{j,k}\sum_{i\neq j}f_{i,j}  \right) p_{j,n},\\
t^0_k& =&{\sum_{n_0=1}^N\sum_{j=1}^J\left(  r^0_{j,k}f_{j,j}+ r^{1}_{j,k}\sum_{i\neq j}f_{j,i} +r^{2}_{j,k}\sum_{i\neq j}f_{i,j}  \right) p_{j,n_0}},\\
\label{auxi4}
s_{k,n}^1&=&\sum_{j=1}^J\left(  \sum_{i\neq j}f_{j,i}   \right) r^{1}_{j,k+K_0}p_{j,n}\\
t_{k}^1&=&{\sum_{n_0=1}^N\sum_{j=1}^J\left(  \sum_{i\neq j}f_{j,i}   \right) r^{1}_{j,k+K_0}p_{j,n_0}},\\
s^2_{k,n}&=&\sum_{j=1}^J\left(  \sum_{i\neq j}f_{i,j}   \right) r^{2}_{j,k+K_0}p_{j,n},\\
\label{auxi5}
t^2_{k}&=&{\sum_{n_0=1}^N\sum_{j=1}^J\left(  \sum_{i\neq j}f_{i,j}   \right) r^{2}_{j,k+K_0}p_{j,n_0}}.
\end{eqnarray}
Again, $s_{k,n}^0$ is weighted average of  SPVs,  and so do $s_{k,n}^1$ and $s_{k,n}^2$.  The variables $t_k^0,t_k^1,$ and $t_k^2$ are introduced to normalize $s_{k,n}^0, s_{k,n}^1,$ and $s_{k,n}^2$.

Define $\bm{z}=[t_1^0,...,t_{K_2}^2,s_{1,1}^0,...,s_{K_2,N}^2,r^0_{1,1},...,r^2_{J,K_2^t}]$, which is a vector consisting of  $L=((K_0 + K_1 + K_2)(N+1)+J(3 K_0+K_1+K_2))$ components. Then problem (\ref{op_8}) can be transformed into a DC programming problem as follows:
\begin{equation}\label{op_9}
\begin{split}
\min_{\bm{z}}\quad  &\sum_{l=1}^{L}\mu_{1,l} z_l \log z_l - \sum_{l=1}^{L}\mu_{2,l} z_l \log z_l    \\
\text{s.t.}\quad & \bm{A}_2\bm{z} = \bm{b}_2, \hspace{1.6cm}(\ref{op_9}.a)\\
&z_l\ge 0, l\in[L].\hspace{1cm} (\ref{op_9}.b)
\end{split}
\end{equation}
In problem (\ref{op_9}), $\bm{\mu}_1=[\mu_{1,1},...,\mu_{1,L}]$ is a vector satisfying $\mu_{1,l}=1$ for $l\le K_0 + K_1 +K_2$ and $\mu_{1,l}=0$ for other values of $l$. In addition, $\bm{\mu}_2=[\mu_{2,1},...,\mu_{2,L}]$ is a vector satisfying $\mu_{2,l}=1$ for $ K_0 + K_1 +K_2<l\le (K_0 + K_1 +K_2)(N+1)$ and $\mu_{2,l}=0$ for other values of $l$. The equality constraint $\bm{A}_2\bm{z} = \bm{b}_2$ comes from  Eqs. $(\ref{op_8}.a)$-$(\ref{op_8}.c)$ and Eqs. (\ref{auxi3})-(\ref{auxi5}). One can see that problems (\ref{op_9}) and (\ref{op_6}) are almost identical apart from the difference in dimension. Thus, Algorithm \ref{dca} can also be used to solve problem (\ref{op_9}).

\begin{algorithm}[t]
    \caption{DCA based Codebook Design for Edge Source Coding with Two Users}
    \label{algo7}
    \renewcommand{\algorithmicrequire}{ \textbf{Input:}} 
    \renewcommand{\algorithmicensure}{ \textbf{Output:}} 
    \begin{algorithmic}[1]
        \REQUIRE    $\bm{p}_1,...,\bm{p_J},\bm{F},K_1^t,K_2^t,\varepsilon$;        \\
        \ENSURE $K_0,\bm{q}_1^0,...,\bm{q}_{K_0}^0,\bm{q}_1^1,...,\bm{q}_{K_1}^1,\bm{q}_1^2,...,\bm{q}_{K_2}^2$;\\
        \FOR{$K_0=1$ to $\min\{K_1^t,K_2^t\}$ }
        \STATE $K_1= K_1^t-K_0, K_2= K_2^t-K_0$;
        \STATE  Calculate $\bm{\mu_1,\mu_2,A_2,b_2}$ according to $\bm{p}_1,...,\bm{p}_J$ and $f_1,...,f_J$;
        \STATE Solve problem (\ref{op_9}) according to Algorithm \ref{dca} and obtain $\bm{z}^{\text{opt}}$;
        \STATE Extract $r^0_{1,1},...,r^2_{J,K_2^t}$ from $\bm{z}^{\text{opt}}$;
        \FOR{$j =1$ to $J$}
        \STATE Find $k_{\max }^* = \arg \max_{k\in [K_*]} r_{j,k}^*$ for $*\in \{0,1,2\}$;
        \STATE Set $r_{j,k_{\max}^* }^* = 1 $ and $r_{j,k}^* = 0 $ for $k\neq k_{\max }^*$ and $*\in \{0,1,2\}$;
        \ENDFOR
        \STATE Compute $\bm{q}_1^0,...,\bm{q}_{K_0}^0,\bm{q}_1^1,...,\bm{q}_{K_1}^1,\bm{q}_1^2,...,\bm{q}_{K_2}^2$ according to Eqs. (\ref{opt_cond_3})-(\ref{opt_cond_5});
        \STATE Compute the objective of problem (\ref{op_7}) $R_{K_0}$ for given $\bm{q}_1^0,...,\bm{q}_{K_0}^0,\bm{q}_1^1,...,\bm{q}_{K_1}^1,\bm{q}_1^2,...,\bm{q}_{K_2}^2$;
        \ENDFOR
        \STATE Return $K_0^* = \arg\min R_{K_0}$ and the corresponding $\bm{q}_1^0,...,\bm{q}_{K_0^*}^0,\bm{q}_1^1,...,\bm{q}_{K_1^*}^1,\bm{q}_1^2,...,\bm{q}_{K_2^*}^2$;
    \end{algorithmic}
\end{algorithm}  

Algorithm \ref{algo7} gives a suboptimal codebook design for the two-user edge source coding. In Algorithm \ref{algo7}, we traverse all possible values of $K_0$. For each $K_0$, the optimal probabilistic clustering scheme is obtained by calling Algorithm \ref{dca}. Afterwards, the probabilities $r^0_{1,1},...,r^2_{J,K_2^t}$ are normalized to be 0-1 and then the optimal codebook designs for this $K_0$ are derived.  The optimal $K_0$ are determined by comparing the transmission costs $R_{K_0}$. 

Before running Algorithm \ref{algo7}, we can have some insights on the results. Qualitatively, the more similar the two users' preferences, the smaller the transmission cost. Note that the user preferences are described by matrix $\bm{F}$. The trace of $\bm{F}$ measures how similar the two users' preferences are. Thus, the higher the trace of $\bm{F}$, the smaller the transmission cost, which will be demonstrated in simulations later.

\subsection{$k$-means++ based Codebook Design for Edge Source Coding with Two Users}

\begin{algorithm}[t]
    \caption{$k$-means++ based Codebook Design for Edge Source Coding with Two Users}
    \label{algo6}
    \renewcommand{\algorithmicrequire}{ \textbf{Input:}} 
    \renewcommand{\algorithmicensure}{ \textbf{Output:}} 
    \begin{algorithmic}[1]
        \REQUIRE    $\bm{p}_1,...,\bm{p_J},\bm{F}$;        \\
        \ENSURE $K_0,\bm{q}_1^0,...\bm{q}_{K_0}^0,\bm{q}_1^1,...\bm{q}_{K_1}^1,\bm{q}_1^2,...\bm{q}_{K_2}^2$;\\
        \FOR{$K_0=1$ to $\min\{K_1^t,K_2^t\}$}
        \STATE Let $K_1 = K_1^t-K_0$ and $K_2 = K_2^t-K_0$;
        \STATE  Randomly select a vector from $\Omega^d$ as $\bm{q}_1^0$ where $\bm{p}_j$ is selected with probability proportional to $\sum_{i=1}^{J}f_{i,j}+\sum_{i=1}^{J}f_{j,i}-f_{j,j} $;
        \FOR{$k=2$ to $K_0$}
        \STATE Compute $G_j=\min_{i<k}D(\bm{p}_j||\bm{q}_i^0)$ for $j=1,...,J$;
        \STATE Randomly select a vector from $\Omega^d$ as $\bm{q}_k$ where $\bm{p}_j$ is selected with probability proportional to $(\sum_{i=1}^{J}f_{i,j}+\sum_{i=1}^{J}f_{j,i}-f_{j,j})G_j^2$;
        \ENDFOR
        \STATE Initialize $\bm{q}_k^1$ and $\bm{q}_k^2$ based on $\Omega^d$ accordingly;
        \STATE Compute $\Omega^0_k,\Omega^1_k,\Omega^2_k$ according to Eqs. (\ref{subset0})-(\ref{subset3});
        \WHILE {at least one of Eqs. (\ref{coup1})-(\ref{coup3}) is FALSE for some $n,k$}
        \STATE  Update $\bm{q}_k^0,\bm{q}_k^1,\bm{q}_k^2$ according to Eqs. (\ref{coup1})-(\ref{coup3});
        \STATE Update $\Omega^0_k,\Omega^1_k,\Omega^2_k$ according to Eqs. (\ref{subset0})-(\ref{subset3});
        \ENDWHILE
        \STATE Compute the objective of problem (\ref{op_7}) $R_{K_0}$ for given $\bm{q}_1^0,...,\bm{q}_{K_0}^0,\bm{q}_1^1,...,\bm{q}_{K_1}^1,\bm{q}_1^2,...,\bm{q}_{K_2}^2$;
        \ENDFOR
        \STATE Return $K_0^* = \arg\min R_{K_0}$ and the corresponding $\bm{q}_1^0,...,\bm{q}_{K_0^*}^0,\bm{q}_1^1,...,\bm{q}_{K_1^*}^1,\bm{q}_1^2,...,\bm{q}_{K_2^*}^2$;
    \end{algorithmic}
\end{algorithm}

In this subsection, we present a variant of the $k$-means++ approach to give a suboptimal codebook design. As discussed in Section \ref{single_user_discrete}, problem (\ref{op_7}) is viewed as a clustering problem and the best clustering scheme is obtained by iteration. 

For fixed $K_0$, let us define 
   \begin{eqnarray}
\label{subset0}
    \Omega_k^0&=&\{j:D(\bm{p}_j||\bm{q}_k^0) = \min_{i\in [K_0]} D(\bm{p}_j||\bm{q}_i^0)  \}, \text{ for } k \in [K_0],\\
    \Omega_k^1&=&\left\{j:D(\bm{p}_j||\bm{q}_k^0) =\min \left\lbrace \min_{i\in [K_0]} D(\bm{p}_i||\bm{q}_i^0), \min_{i\in [K_1]} D(\bm{p}_i||\bm{q}_i^1) \right\rbrace \right\},  \text{ for } k \in [K_0],\\
    \Omega_{k+K_0}^1&=&\left\{j:D(\bm{p}_j||\bm{q}_k^1) =\min \left\lbrace \min_{i\in [K_0]} D(\bm{p}_i||\bm{q}_i^0), \min_{i\in [K_1]} D(\bm{p}_i||\bm{q}_i^1) \right\rbrace \right\},  \text{ for } k \in [K_1],\\  
        \Omega_k^2&=&\left\{j:D(\bm{p}_j||\bm{q}_k^0) =\min \left\lbrace \min_{i\in [K_0]} D(\bm{p}_i||\bm{q}_i^0), \min_{i\in [K_2]} D(\bm{p}_i||\bm{q}_i^2) \right\rbrace \right\},  \text{ for } k \in [K_0],\\
        \label{subset3}
    \Omega_{k+K_0}^2&=&\left\{j:D(\bm{p}_j||\bm{q}_k^2) =\min \left\lbrace \min_{i\in [K_0]} D(\bm{p}_i||\bm{q}_i^0), \min_{i\in [K_2]} D(\bm{p}_i||\bm{q}_i^2) \right\rbrace \right\},  \text{ for } k \in [K_2].
\end{eqnarray}
It can be seen that $\Omega_k^0$ are disjoint sets and $\cup_{k}\Omega_k^0=\Omega^d$, and so do $\Omega_k^1$ and $\Omega_k^2$ . When the two users' requests are identical, $q_k^0$ is used to encode content items in $\Omega^0_k$. When the two users' requests are different, $q_k^0$ is used to encode the content item requested by user 1 and in $\Omega^0_k$. In addition, $\bm{q}_k^1$ is used to encode the content item requested by user 1 and in $\Omega^0_{k+K_0}$. The sets $\Omega^2_k$ and $\Omega^2_{k+K_0}$ have similar meanings. 
The objective of problem (\ref{op_7}) can be rewritten as 
\begin{equation}\label{trans_cost3}
\begin{split}
   R(\bm{q}_k^0,\bm{q}_k^1,\bm{q}_k^2,\Omega^0_k,\Omega^1_k,\Omega^2_k)=&\sum_{k=1}^{K_0}\left( \sum_{j\in \Omega_k^0}f_{j,j} D(\bm{p}_j||\bm{q}_k^0) +\sum_{j\in \Omega_k^1}\left( \sum_{i\neq j}f_{j,i}\right)  D(\bm{p}_j||\bm{q}_k^0)\right.\\&\left.+\sum_{j\in \Omega_k^2}\left( \sum_{i\neq j}f_{i,j}\right)  D(\bm{p}_j||\bm{q}_k^0)\right) + \sum_{k=1}^{K_1}\sum_{j\in \Omega_{k+K_0}^1}\left( \sum_{i\neq j}f_{j,i}\right)  D(\bm{p}_j||\bm{q}_k^1)\\
   &+ \sum_{k=1}^{K_2}\sum_{j\in \Omega_{k+K_0}^2}\left( \sum_{i\neq j}f_{i,j}\right)  D(\bm{p}_j||\bm{q}_k^2).
\end{split}
\end{equation}
It should be noted that the sets $\Omega^0_k,\Omega^1_k,$ and $\Omega^2_k$ rely on the values of $\bm{q}^0_k,\bm{q}^1_k, $ and $\bm{q}^2_k$.

 Eq. (\ref{trans_cost3}) enables us to improve the codebook design $\bm{q}^0_k,\bm{q}^1_k, \bm{q}^2_k$ through $\Omega^0_k,\Omega^1_k,\Omega^2_k$. 
Taking partial derivatives of Eq. (\ref{trans_cost3}) with respect to $\bm{q}^0_k,\bm{q}^1_k, \bm{q}^2_k$ yields
\begin{eqnarray}
\label{coup1}
    q_{k,n}^0 &=& \frac{ \sum_{j\in \Omega_k^0}f_{j,j} p_{j,n} \!+\!\sum_{j\in \Omega_k^1}\left( \sum_{i\neq j}f_{j,i}\right)  p_{j,n}\!+\!\sum_{j\in \Omega_k^2}\left( \sum_{i\neq j}f_{i,j}\right)  p_{j,n}}{\sum_{n_0=1}^{N}\left( \sum_{j\in \Omega_k^0}f_{j,j} p_{j,n_0} \!+\!\sum_{j\in \Omega_k^1}\left( \sum_{i\neq j}f_{j,i}\right)  p_{j,n_0}\!+\!\sum_{j\in \Omega_k^2}\left( \sum_{i\neq j}f_{i,j}\right)  p_{j,n_0}\right) },\\
    \label{coup2}
        q_{k,n}^1 &=& \frac{ \sum_{j\in \Omega_{k+K_0}^1}\left( \sum_{i\neq j}f_{j,i}\right)  p_{j,n}}{\sum_{n_0=1}^{N}\sum_{j\in \Omega_{k+K_0}^1}\left( \sum_{i\neq j}f_{j,i}\right)  p_{j,n_0}},\\
        \label{coup3}
    q_{k,n}^2 &=& \frac{ \sum_{j\in \Omega_{k+K_0}^2}\left( \sum_{i\neq j}f_{i,j}\right)  p_{j,n}}{\sum_{n_0=1}^{N}\sum_{j\in \Omega_{k+K_0}^2}\left( \sum_{i\neq j}f_{i,ji}\right)  p_{j,n_0}}.
\end{eqnarray}
Eqs. (\ref{coup1})-(\ref{coup3}) minimize $R(\bm{q}_k^0,\bm{q}_k^1,\bm{q}_k^2,\Omega^0_k,\Omega^1_k,\Omega^2_k)$ for fixed $\Omega^0_k,\Omega^1_k,\Omega^2_k$ and further reveal  the coupling between  $\Omega^0_k,\Omega^1_k,\Omega^2_k$ and $\bm{q}^0_k,\bm{q}^1_k, \bm{q}^2_k$. Combining Eqs. (\ref{subset0})-(\ref{subset3}) and (\ref{coup1})-(\ref{coup3}), we propose a $k$-means++ based algorithm to give a heuristic codebook design in Algorithm \ref{algo6}.

Algorithm \ref{algo6} traverses  all the possible values of $K_0$ and compares the transmission costs they bring. For each value of $K_0$, Algorithm \ref{algo6} achieves a locally optimal solution by improving the randomly initialized one iteratively. The fact that each iteration reduces the transmission cost ensures the convergence of Algorithm \ref{algo6}.

 \section{Simulation Results}\label{section7}

 In this section, we present numerical results to demonstrate the performance of the proposed algorithms and the potential of edge source coding. Throughput this section, we assume $L=20$, i.e., each content item consists of 20 symbols. We always set $\varepsilon=10^{-6}$ in the simulations.
 
 \begin{figure}
     \centering
     \includegraphics[width=8cm]{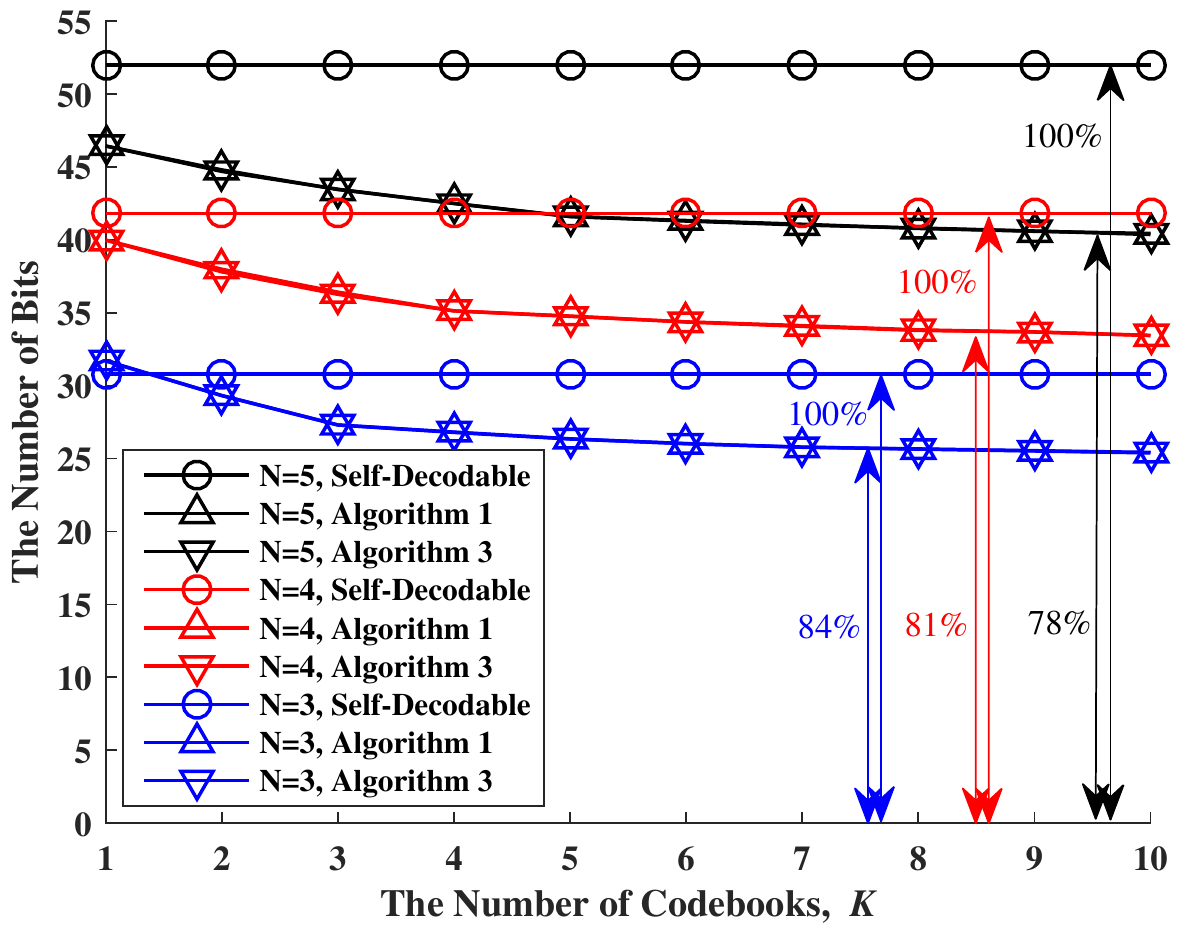}
     \caption{Expected number of bits versus the number of codebooks under  discrete user preferences.}    
     \label{fig1}  
 \end{figure}
 
 Fig. \ref{fig1} presents the expected number of bits to represent a requested content item versus the number of codebooks under  discrete user preferences, where the total number of SPVs with nonzero probability is set to be $J=1000$. These SPVs are uniformly sampled from $\Omega$. It is not surprising that the larger the alphabet, the greater the number of bits needed  to represent a requested content item. To demonstrate the potential of edge source coding, we compare our algorithms with a traditional self-decodable method, in which each requested content item is encoded according to its SPV and carries a codebook for  decoding. Thus, the self-decodable method encode a content item with SPV $\bm{p}$ into $LH(\bm{p}) +  \sum_{n=1}^{N}\log p_n$ bits. In Fig. \ref{fig1}, it can be seen the proposed  DCA based and $k$-means++ based algorithms, i.e., Algorithms \ref{dca_clustering} and \ref{algo1}, always achieve similar performance. Except the case $N=3$ and $K=1$, Algorithms  \ref{dca_clustering} and \ref{algo1}   use fewer bits to represent a requested content item than the self-decodable method does. When $N=5$ and $K=10$, our algorithms can even save 22\% of the total number of bits. For fixed $N$, the number of bits to represent a requested content item decreases with an increase in the number of codebooks. This is because the mismatch between SPVs and  codebooks is small when more codebooks are used. However, an increase in the number of codebooks $K$ only slightly reduces the number of bits when $K$ is large. 

\begin{figure}[t] 
    \centering
    \subfigure[Uniform]{
        \begin{minipage}[t]{0.5\linewidth}
            \centering
            \includegraphics[width=2.8in]{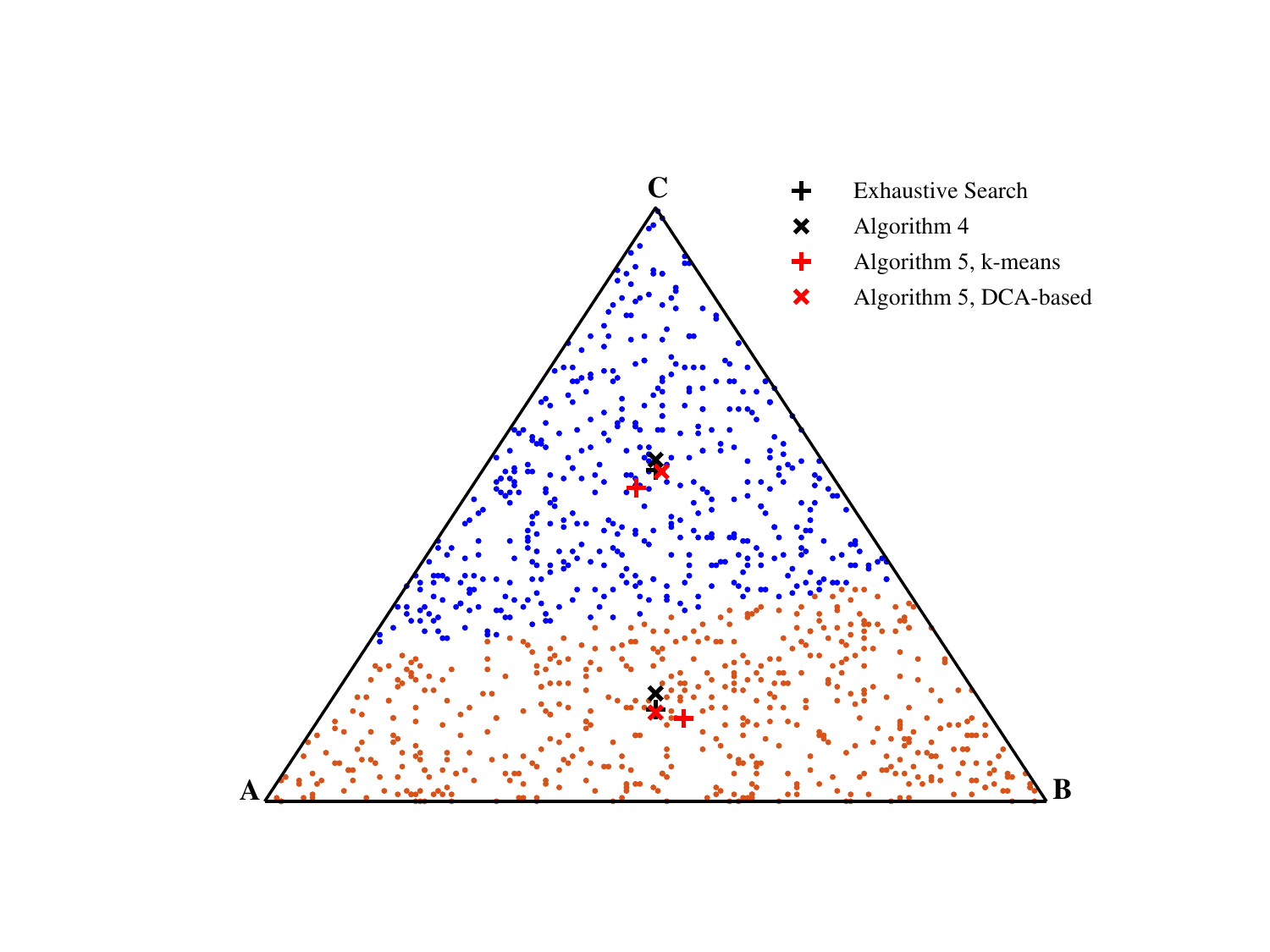}
        \end{minipage}
    }%
    \subfigure[Nonuniform]{
        \begin{minipage}[t]{0.5\linewidth}
            \centering
            \includegraphics[width=2.8in]{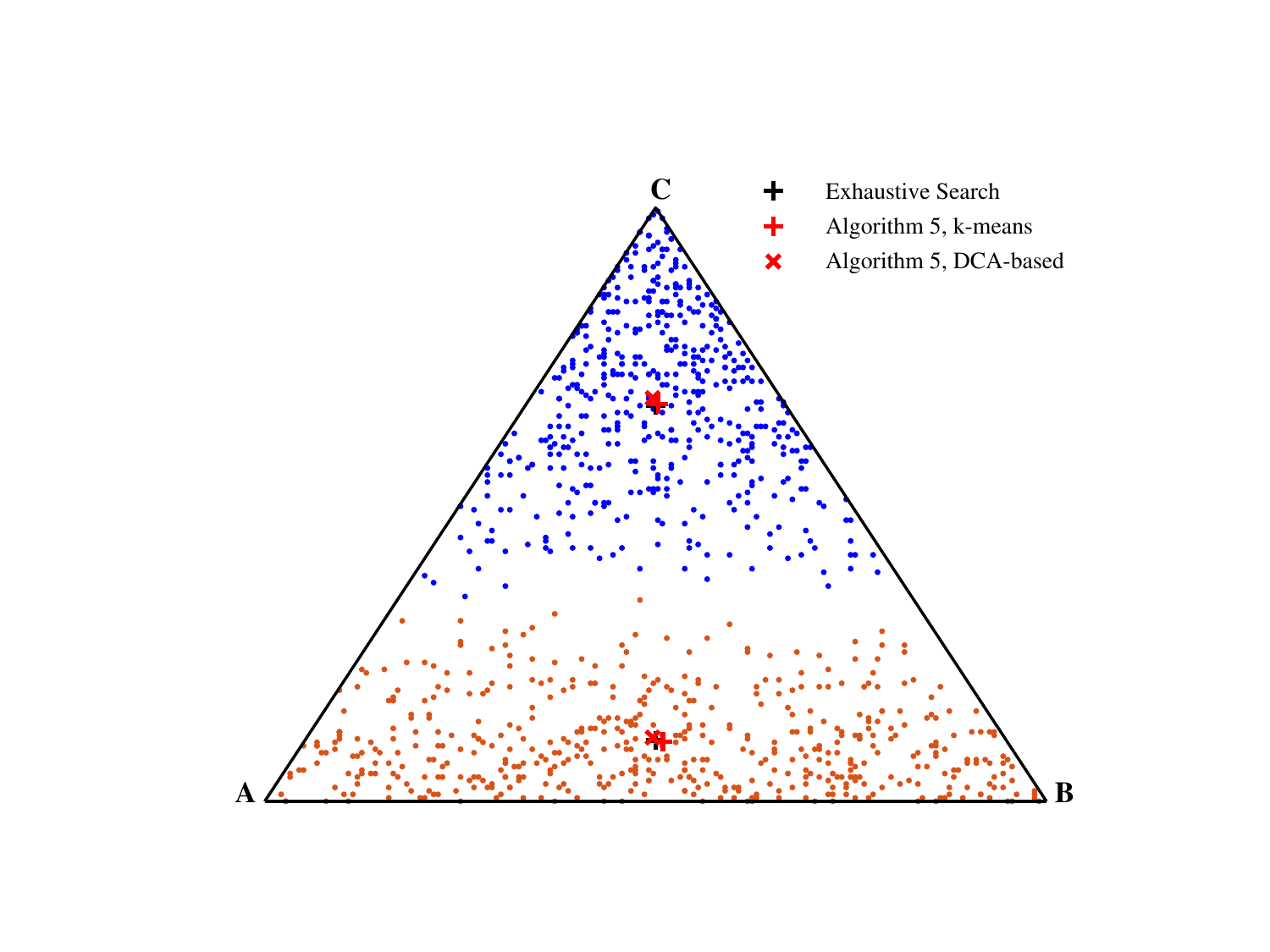}
        \end{minipage}
    }
\centering
\caption{Edge source coding under uniform and nonuniform user preferences}   \label{fig23}
\end{figure}

 Fig. \ref{fig23}  illustrates the performance of Algorithms \ref{algo2} and \ref{algo3}. For the sake of illustration, we set $N=3$ and $K=2$ in the simulation. The notation $A,B,$ and $C$ in these two figures has identical meaning to that in Fig. \ref{n=3k=2}. In the subfigures (a) and (b), the function $f$ is  respectively set to be $f=1/{\int_{\bm{p}\in \Omega}d\bm{p}}$ and $f=(||\bm{p}-\bm{p}_0||_2)/{\int_{\bm{p}\in \Omega}||\bm{p}-\bm{p}_0||_2d\bm{p}}$, where $\bm{p}_0 = [1/3,1/3,1/3]$ and $||\cdot||_2$ denotes the Euclidean norm. Thus, the user preferences are uniform and nonuniform respectively in the two subfigures. The black cross and  x mark represent the solutions given by exhaustive search and Algorithm $\ref{algo2}$, respectively. The red  x mark and cross represent the solutions when Algorithm \ref{algo3} calls  Algorithms \ref{dca_clustering} and \ref{algo1}, respectively. The blue and orange dots represents the sample points when Algorithm \ref{algo3} runs ($S=1000$ points are sampled).  The color of a dot indicates the codebook used by this dot. Fig. \ref{fig23}  reveals that the optimal codebook designs vary with user preferences. It is seen that the codebook designs given by Algorithms \ref{algo2} and \ref{algo3} are very close to the optimal solution given by  exhaustive search. For the uniform user preferences, the expected numbers of bits resulting from Algorithm \ref{algo2},  Algorithm \ref{algo3} calling Algorithm \ref{dca_clustering}, Algorithm \ref{algo3} calling Algorithm \ref{algo1},  and exhaustive search are 29.0682, 29.0486, 29.0816,   and   29.0457, respectively. For the considered nonuniform user preferences, the expected numbers of bits resulting from   Algorithm \ref{algo3} calling Algorithm \ref{dca_clustering}, Algorithm \ref{algo3} calling Algorithm \ref{algo1}, and exhaustive search are 26.2575, 26.2753,  and 27.2523. The gaps between the expected numbers of bits given by our algorithms   and the optimal are very small. 
 
 \begin{figure}
    \centering
    \includegraphics[width=8cm]{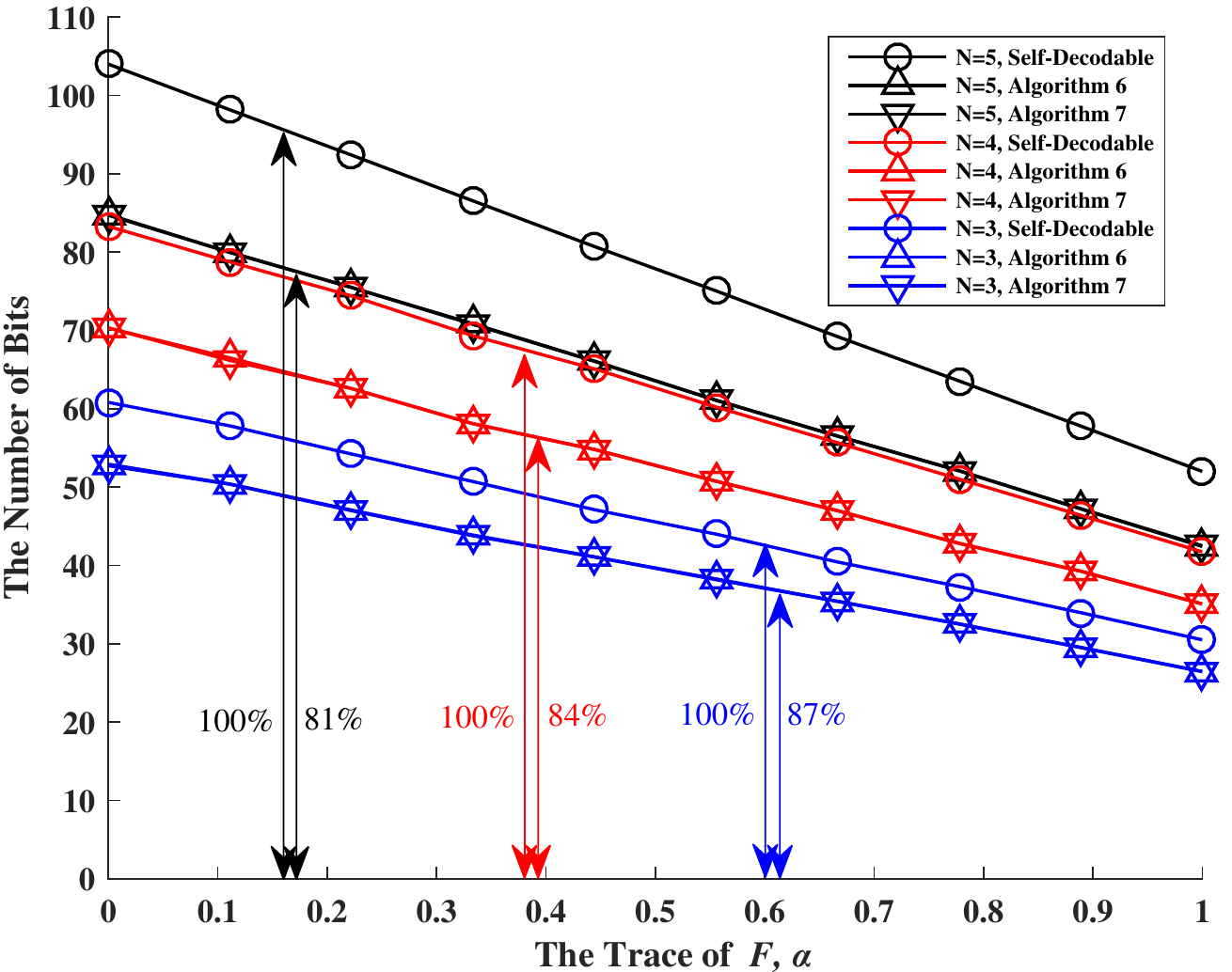}
    \caption{Expected number of bits versus the trace of joint  preference matrix for edge source coding  with two users.}    
    \label{fig4}  
\end{figure}

Fig. \ref{fig4} presents the number of bits to represent a requested content item versus the similarity of the user preferences for edge source coding  with two users. We assume there are $J=1000$ content items which are with SPVs uniformly sampled from $\Omega$. The joint preference matrix $\bm{F}=(f_{i,j})_{J\times J}$ is set to be 
\begin{equation}
    f_{i,j} = \begin{cases}
   \frac{\alpha}{J}, &\text{if } i =j;\\
    \frac{1-\alpha}{J(J-1)}, &\text{others.}
    \end{cases}
\end{equation}
Then, the trace of $\bm{F}$ is $\alpha$, which describes the similarity of the two users' preferences. In Fig. \ref{fig4}, we assume the total numbers of codebooks used by user 1 and user 2 are $K_1^t=4$ and $K_2^t=4$. Again, we compare our algorithms with the self-decodable method, in which each content item is attached with a codebook in transmission.  It can be seen that the expected number of bits to satisfy user requests almost linearly decrease with the trace of $\bm{F}$. This is due to the fact that higher trace of $\bm{F}$ induces more multicasting opportunities. The   DCA based and $k$-means++ based  algorithms, i.e., Algorithms \ref{algo7} and \ref{algo6}, have similar performance. Furthermore, our algorithms always achieve lower number of bits than the self-decodable method does. When $N=5$, our algorithms can save 19\% of the total number of bits. 

 \begin{figure}
    \centering
    \includegraphics[width=8cm]{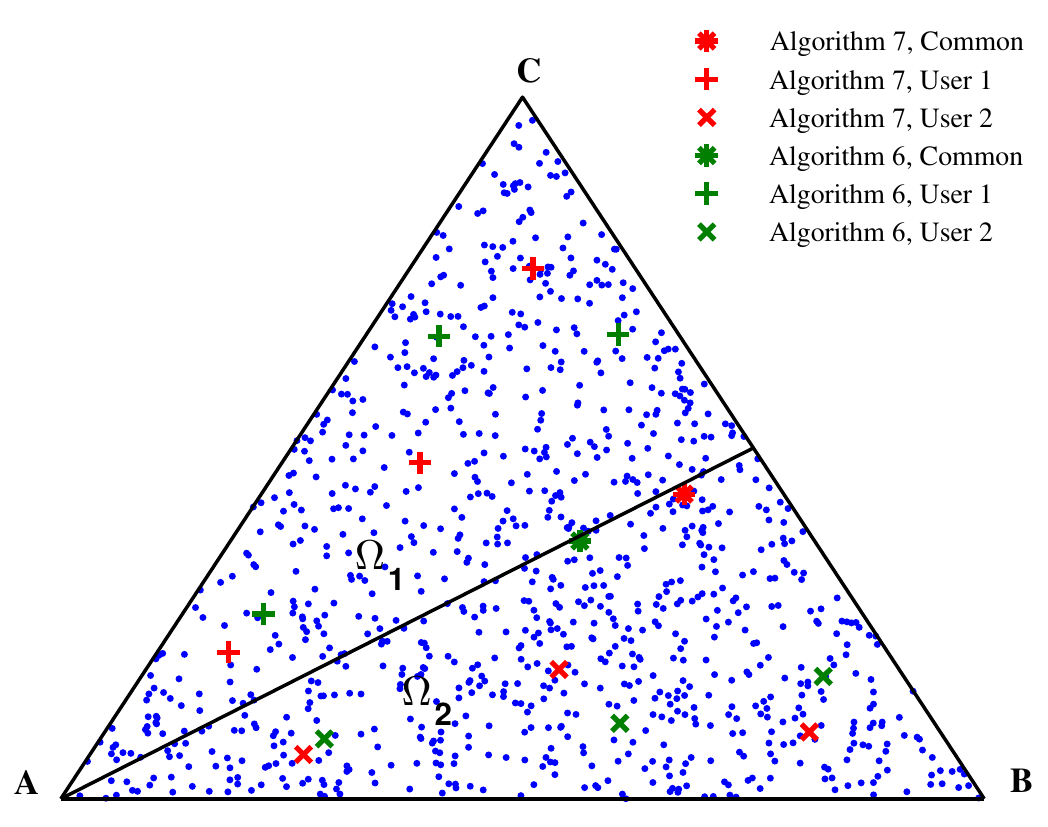}
    \caption{Expected number of bits versus the similarity of user preferences in two-user case.}    
    \label{fig5}  
\end{figure}
Fig. {\ref{fig5}} illustrates optimal codebook designs for $N=3$ and $K_1^t=K_2^t= 4$. There are total $J=1000$ content items.   We assume user 1 is interested only in content items with SPVs in $\Omega_1$ and user 2 is interested only in content items with SPVs in $\Omega_2$.  The green and red asterisks represents the common codebooks given by Algorithms \ref{algo7} and \ref{algo6}, respectively.  The cross and x mark represents the codebooks exclusively used by user 1 and user 2, respectively. It can be seen both the common codebooks are located around the boundary between $\Omega_1$ and $\Omega_2$. The codebook designs resulting from the two algorithms differ sharply, which implies the existence of multiple locally optimal solutions. The numbers  of bits given by Algorithms \ref{algo7} and \ref{algo6} are   51.0691 and 50.9569. Again, the performances of the two algorithms are similar. 

\section{Conclusion and  Directions for Future Work}\label{section8}

In this paper, we have studied user preference aware lossless data compression at the edge, where content items are characterized by SPVs and  user preferences are described by probability density functions.  We have noted that  the statistical distributions of symbols over the links generally differ from that in the information source and have presented an edge source coding method, which employs finitely many codebooks to encode the requested content items at the service provider. An optimization problem has been formulated to give the optimal codebook design for edge source coding, which is however nonconvex in general. The optimal solution for  edge source coding with one codebook has been given by the method of Lagrange multipliers. Furthermore, an optimality condition has been presented. For   edge source coding  under discrete user preferences, the optimization problem reduces to a clustering problem.  DCA based and $k$-means++   algorithms have been proposed to give   codebook designs. For   edge source coding under continuous user preferences, an iterative algorithm has been proposed to give a codebook design according to the optimality condition. To provide a codebook design with low computational complexity, a sampling based algorithm has also been proposed. In addition, we have extended edge source coding to the two-user case and codebooks have been designed to utilize the multicasting opportunities. It has been shown that the solutions of the proposed algorithms are close to the optimal. In the two-user case, the expected number of bits needed to represent a requested content item linearly decreases with the trace of the joint preference matrix. Both theoretical analysis and simulations have verified that the optimal codebook design in edge source coding relies on user preferences. 

Significant future topics include edge source coding with more than two users, more-refined theoretical analysis on the optimal codebook design and the resulting transmission cost, as well as practical user preference model based on real data. 

\begin{appendix}
  Here we show the transformation 
  from the optimal solution of problem (\ref{op_4}) to that of problem (\ref{op_2}).   To show that, we only need to show the existence of 0-1 optimal solution of problem (\ref{op_4}). The core idea is to improve the solution by removing elements unequal to 0 or 1. Suppose $\{r_{j,k}^*: j \in [J], k\in[K]\}$ is an optimal solution of problem (\ref{op_4}) and let $R_{\ref{op_2}}(r_{j,k}^*)$ denote the corresponding optimization value.  We have 
  \begin{equation}
      \sum_{k=1}^{K}r_{j,k}^* = 1,
  \end{equation}
  for $j\in[J]$. From Eq. (\ref{opt_cond3}), we can obtain $\{\bm{q}_k^*:k\in[K]\}$ such that $\bm{q}_k^*$ and $r_{j,k}^*$ form the optimal solution of problem (\ref{op_3}). Let $R_{\ref{op_3}}(\bm{q}^*_k,r_{j,k}^*)$  denote the optimization value of problem (\ref{op_3}). Then we have $R_{\ref{op_3}}(\bm{q}^*_k,r_{j,k}^*)=R_{\ref{op_2}}(r_{j,k}^*)$.  If there exist $j_0$ such that $0<r_{j_0,k}^*<1$ for certain values of $k$, we set 
  \begin{equation}\label{remove}
      r_{j,k}^{**} = \begin{cases}
      r_{j,k}^*, & \text{if } j \neq j_0,\\
      1, & \text{if } j = j_0 \text{ and } k = \arg\min_{k_0\in [K]} D(\bm{p}_j||\bm{q}_{k_0}^*),\\
      0, &  \text{if } j = j_0 \text{ and } k \neq  \arg\min_{k_0\in [K]} D(\bm{p}_j||\bm{q}_{k_0}^*).
      \end{cases}
  \end{equation}
  Then, $\{r_{j,k}^{**}\}$ is feasible for problem ({\ref{op_4}}). In addition, we have
  \begin{equation}
     R_{\ref{op_2}}(r_{j,k}^{**})= R_{\ref{op_3}}(\bm{q}^{**}_k,r_{j,k}^{**})\ge R_{\ref{op_3}}(\bm{q}^{*}_k,r_{j,k}^{**})\ge R_{\ref{op_3}}(\bm{q}^{*}_k,r_{j,k}^{*}) =  R_{\ref{op_2}}(r_{j,k}^{*}),
  \end{equation}
  where $\bm{q}_k^{**}$ is derived from Eq. ({\ref{opt_cond3}}) by substituting $\{r_{j,k}^{**}\}$.
  As a result, $\{r_{j,k}^{**}\}$  is an optimal solution of problem ({\ref{op_4}}). By applying Eq. (\ref{remove}) iteratively, we can remove all the non 0-1 elements of an optimal solution of problem ({\ref{op_4}}) without the loss of optimality. 
  
\end{appendix}

\linespread{0.89}
\bibliographystyle{IEEEtran}

\end{document}